\documentclass[10pt]{article}
\usepackage[utf8]{inputenc}
\usepackage[margin=1in]{geometry}
\usepackage[normalem]{ulem}
\usepackage{amsthm}
\usepackage{bm}
\usepackage{mathtools}
\usepackage{thmtools}
\declaretheoremstyle[headfont=\normalfont]{normalhead}
\usepackage{color}
\usepackage{graphicx}
\usepackage{morefloats}
\usepackage{hyperref}
\usepackage{ amssymb }
\usepackage{textcomp}
\usepackage{url}
\usepackage{float}
\usepackage{appendix}
\usepackage{authblk}
\usepackage{caption}
\usepackage[
backend=biber,
giveninits=true,
uniquename=false,
date=year,
style=numeric-comp,
uniquelist=false,
sorting=none,
url=false,
isbn=false,
doi=false,
maxbibnames=10,
minbibnames=1,
maxcitenames=3,
mincitenames=1
]{biblatex}
\DeclareSourcemap{
  \maps[datatype=bibtex]{
    \map{
      \pertype{article}
      \step[fieldsource=journaltitle, match=\regexp{.+}]
      \step[fieldset=eprint,       null]
      \step[fieldset=eprinttype,   null]
      \step[fieldset=eprintclass,  null]
    }
    \map{
      \pertype{inproceedings}
      \step[fieldsource=booktitle, match=\regexp{.+}]
      \step[fieldset=eprint,       null]
      \step[fieldset=eprinttype,   null]
      \step[fieldset=eprintclass,  null]
    }
    \map{
      \pertype{conference}
      \step[fieldsource=booktitle, match=\regexp{.+}]
      \step[fieldset=eprint,       null]
      \step[fieldset=eprinttype,   null]
      \step[fieldset=eprintclass,  null]
    }
  }
}

\newcommand\R{\mathbb{R}}

\newcommand{\calP}{\mathcal{P}}

\newcommand{\ra}{\rightarrow}

\renewcommand\epsilon{\varepsilon}

\DeclareMathOperator*{\argmin}{arg\,min}

\newtheoremstyle{mydef}
{\topsep}{\topsep}%
{}{}%
{\itshape}{}
{\newline}
{%
  \rule{\textwidth}{0.0pt}\\*%
  \thmname{#1}~\thmnumber{#2}\thmnote{\-\ #3}.\\*[-1.5ex]%
  \rule{\textwidth}{0.0pt}}%



\begin{document}

\newtheorem{conjecture}{Conjecture}
\newtheorem{proposition}{Proposition}
\newtheorem{theorem}{Theorem}[section]
\newtheorem{definition}{Definition}[section]
\newtheorem{question}{Question}
\newtheorem{remark}{Remark}
\newtheorem{proposal}{Proposal}
\newtheorem{lemma}{Lemma}[section]
\newtheorem{corollary}{Corollary}[section]
\newtheorem{observation}{Observation}[section]
\newtheorem{assumption}{Assumption}[section]

\renewcommand\appendixpagename{Supplementary materials}

\date{\today}

\title{Biophysics-informed deep operator learning for inverse problems with application to electrophysiological source reconstruction}

\author[1]{Eardi Lila}
\author[2,3]{Erica R. Peterson}
\author[2]{Alexis N. Bosseler}
\author[4,5]{J. Nathan Kutz}
\author[2,6]{Samu Taulu}

\affil[1]{Department of Biostatistics, University of Washington, Seattle, WA, USA}
\affil[2]{Institute for Learning \& Brain Sciences, University of Washington, Seattle, WA, USA}
\affil[3]{Department of Speech \& Hearing Sciences, University of Washington, Seattle, WA, USA}

\affil[4]{Department of Applied Mathematics, University of Washington, Seattle, WA, USA}
\affil[5]{Department of Electrical and Computer Engineering, University of Washington, Seattle, WA, USA}
\affil[6]{Department of Physics, University of Washington, Seattle, WA, USA}

\maketitle
\begin{abstract}
Electrophysiological brain signals are typically acquired through indirect and noisy measurements, providing transformed representations of the underlying neural activity. Source reconstruction---the inverse problem of resolving underlying neural signals from these measurements---is essential for mapping brain function but remains challenging because it is ill-posed and sensitive to noise. Deep learning methods have shown promise across a range of inverse problems, yet many do not explicitly incorporate the biophysical principles governing data generation, limiting data efficiency and adaptation across subjects. Here, we introduce DeepOp-Informed, a biophysics-informed geometric deep operator learning framework that embeds the biophysics of the sensing process into the model through a custom differentiable layer, enabling more efficient learning and improved reconstruction performance. This layer enables the neural network to adapt to subject-specific variations in the physics of signal generation, resulting from differences in brain anatomy and sensor positioning. In realistic magnetoencephalography simulations, DeepOp-Informed generalizes to forward models from held-out subjects, reducing reconstruction error relative to several neural-network and classical baselines. Applied to adolescent auditory-evoked recordings, it produces anatomically plausible reconstructions localized to the auditory cortex. While our application focuses on magnetoencephalography, the framework is general and may be adaptable to other imaging modalities.

\end{abstract}


\noindent

\noindent

\section{Introduction}
Source reconstruction refers to the task of mapping electrophysiological measurements onto the brain to infer the spatial distribution of neuronal activity. Magnetoencephalography (MEG) is one of the primary imaging modalities for measuring electrophysiological signals, offering high temporal resolution and providing recordings of the magnetic field generated by neuronal currents. \parencite{baillet2001electromagnetic, baillet2017magnetoencephalography}. However, source reconstruction is a highly ill-posed inverse problem, meaning that small perturbations in the sensor-level data can lead to large deviations in the reconstructed source signal, making the process highly susceptible to noise. Standard approaches tackle this problem by first constructing a \textit{forward model} that relates neural currents to the externally measured magnetic field under the quasi-static approximation to Maxwell’s equations \parencites{sarvas1987basic,hamalainen1993magnetoencephalography,taulu2021unified}, incorporating the physics of the sensing mechanism, head anatomy, and sensor geometry. This step is followed by a regularized least-squares reconstruction of the latent source signal \parencites[see, e.g., ][]{dale1993improved, hamalainen1994interpreting, mosher1998recursive, pascual-marqui2002standardized, vanveen1997localization, dale2000dynamic}. However, these approaches may lack spatial accuracy, particularly in low signal-to-noise (SNR) settings, due to their overly simplistic modeling assumptions \parencite{becker2015brainsource}. Specifically, most of these methods are linear in the data, meaning that the reconstructed source at each location is restricted to be a linear combination of the sensor-level measurements. Although nonlinear extensions do exist, specifying and tuning the prior information in such models remains challenging.

Neural networks have recently emerged as a promising alternative to these traditional methods in both MEG and EEG (electroencephalography) \parencites{jiao2022graph, jiao2024multimodal, wang2024fast, sun2022deep}, as they can represent more complex source-sensor relationships. These approaches generate synthetic source data, apply a forward model to obtain corresponding M/EEG sensor measurements, and use neural networks to learn an approximation of the inverse mapping. The equations governing the relationship between neural currents and the measured magnetic field are indirectly learned by the neural networks from the provided source-sensor data pairs. However, learning these mechanistic aspects implicitly is typically data-inefficient and can produce unstable reconstruction models, as demonstrated by our simulation studies and the broader inverse-problems literature \parencite{antun2020instabilities}.

In this work, we introduce DeepOp-Informed, a biophysics-informed architecture that incorporates external mechanistic information by conditioning on each sample’s forward model through a differentiable regularized inverse layer (Figure~\ref{fig:intro}). DeepOp-Informed is designed to learn a reconstruction operator, $\mathcal R_{\mathcal{K},\theta}: \mathbb{R}^s \rightarrow \mathbb{R}^p$, that maps the sensor space $\mathbb{R}^s$, with $s$ sensors, to the source space $\mathbb{R}^p$, with $p$ dense cortical dipoles. The reconstruction operator depends both on the learnable neural-network parameters $\theta$ and on a user-provided forward model $\mathcal{K} \in \mathbb{R}^{s \times p}$, which encodes the biophysical mapping from sources to observed signals; for this reason, the approach is termed biophysics-informed. We use the term \textit{operator} because the formulation naturally generalizes beyond the MEG setting to mappings between more general spaces, linking the proposed framework to the deep-operator-learning literature \parencites{lu2021learning, kovachki2024operator, reinhardt2024statistical}. 

Specifically, DeepOp-Informed consists of two modules: a differentiable biophysics-informed custom layer and a geometric deep learning refinement module. The biophysics-informed layer produces an initial reconstruction by combining biophysical knowledge, encoded through a provided forward model $\mathcal{K}$, with a learned regularization term. The specific form of this combination draws inspiration from empirical Bayes methods \parencites{efron2009empirical, saremi2019neural}. The geometric deep learning module then refines this initial solution by leveraging a graph-based representation that encodes cortical spatial structure. All learnable components of the two modules are trained end-to-end on a curated synthetic dataset, allowing the model to be robust to non-idealities, such as realistic noise, and to incorporate prior knowledge about the nature of signals to be reconstructed.

In contrast to prior neural-network-based approaches to the MEG inverse problem \parencites{jiao2022graph, jiao2024multimodal, wang2024fast, sun2022deep}, DeepOp-Informed conditions the network on the forward model $\mathcal{K}$, explicitly embedding the physics of the sensing process into the architecture. This biophysics-informed architecture has two main advantages. First, it enables the model to directly incorporate mechanistic information from the forward model, improving data efficiency. Second, once the parameters $\theta$ have been learned, the same network can be deployed across subjects by providing a subject-specific forward model $\mathcal{K}$, without retraining. Classical source reconstruction methods for electrophysiological source imaging, such as minimum-norm estimate \parencite{hamalainen1994interpreting}, also implicitly define a reconstruction operator $\mathcal{R}_{\mathcal{K},\theta}$ through the solution of a least-squares problem. However, most of these methods are linear and differ mainly in how the source covariances are specified and estimated \parencite{mosher2003equivalence}. As a result, their reconstruction performance may be inherently limited, and they cannot take advantage of the expressive power of neural networks to model nonlinear structure in the data. Neural-network approaches to image reconstruction have also been developed for magnetic resonance imaging and computed tomography \parencite{argyrou2012tomographic, oh2018eternet, schlemper2018deepa, zhu2018image}, including deep unrolled architectures that explicitly incorporate information from the ``forward model'' \parencite{stuart2010inverse, arridge2019solving, aggarwal2019modl}. However, these approaches are typically designed for settings with higher SNRs and are generally not evaluated in settings where the forward model is subject-specific, as in our case, or where the source space is a nonlinear surface.

\begin{figure}[!htb]%
\centering
\includegraphics[width = 0.95\textwidth]{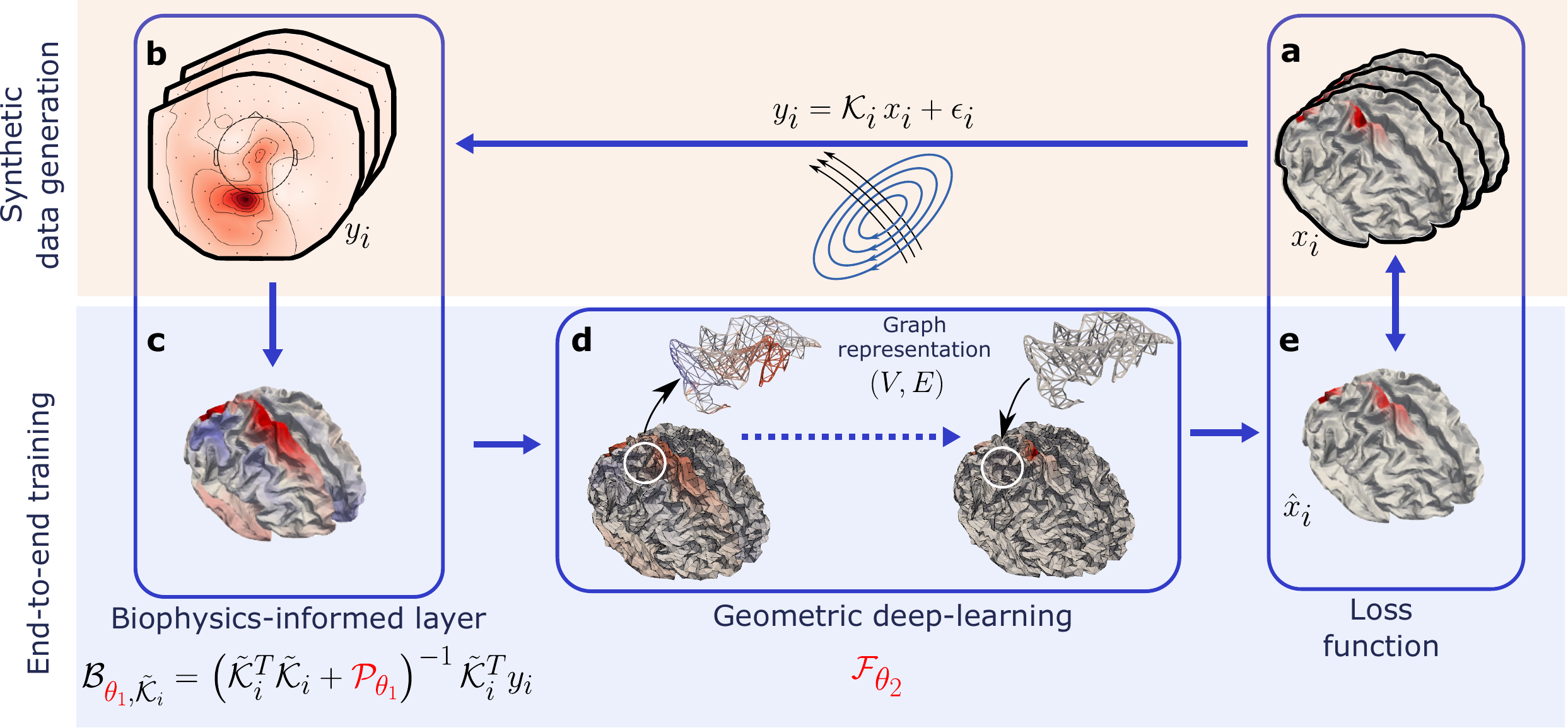}
\caption{\textit{Illustration of DeepOp-Informed in the context of MEG reconstruction.} \textbf{a.} We use a reference cortical surface represented as a triangular mesh whose vertices define the locations of $p$ dipoles oriented normal to the surface. We assume all subjects’ cortical surfaces have been registered to this reference cortical surface. We then use the reference cortical surface mesh to generate synthetic source signals $x_i \in \R^p$ by generating a random number of equal-intensity point sources centered at randomly selected vertices on the cortical surface, followed by smoothing to produce spatially localized but more diffuse peaks. \textbf{b.} Given the positions of the $s$ extracranial sensors and a head model, we construct the corresponding subject-specific forward models, $\mathcal{K}_i \in \mathbb{R}^{s \times p}$. We then map the synthetic source signals to sensor space using the subject-specific forward models and add realistic noise to obtain $y_i \in \R^s$. DeepOp-Informed is then trained on the synthetic dataset of triplets $(y_i, x_i, \tilde{\mathcal K}_i)$ to learn to reconstruct $x_i$ from $y_i$ and $\tilde{\mathcal K}_i$. Here, $\tilde{\mathcal K}_i$ is a surrogate forward model that may be taken to be equal to $\mathcal K_i$, or it may be used to coarsen the reconstruction or to downweight regions deemed a priori less likely to activate. \textbf{c.} The biophysics-informed input layer of DeepOp-Informed, $\mathcal B_{\theta_1,\tilde{\mathcal K}_i}:\R^s \rightarrow \R^p$ explicitly combines mechanistic knowledge, encoded in the surrogate forward model $\tilde{\mathcal K}_i$, with a learned component $\mathcal P_{\theta_1} \in \R^{p \times p}$. The integration of the learned component with the forward model is inspired by the structure of classical linear reconstruction methods, or equivalently by the form of $\mathbb{E}\!\left[x_i \mid y_i, \tilde{\mathcal K}_i \right]$ under joint Gaussian assumptions on $x_i$ and $\epsilon_i$, but with a learned source covariance structure parameterized by its inverse $\mathcal P_{\theta_1}$ (see Section~\ref{sec:methods}). \textbf{d.} The output of this input layer, a $p$-dimensional signal defined on the cortical mesh, is then passed to a graph convolutional network module $\mathcal F_{\theta_2}:\R^p \rightarrow \R^p$, which further refines the estimate by leveraging the geometry of the cortical mesh, represented as a graph $(V,E)$ constructed from the mesh vertices and triangle edges. \textbf{e.} All data-driven components of the model, shown in red, are learned jointly by minimizing the mean squared error $\frac{1}{n}\sum_{i=1}^n \| x_i - \mathcal R_{\theta,\tilde{\mathcal K}_i}(y_i) \|_2^2$ over $\theta = (\theta_1, \theta_2)$, where $\mathcal R_{\theta,\tilde{\mathcal K}_i} = \mathcal F_{\theta_2}\circ \mathcal B_{\theta_1,\tilde{\mathcal K}_i}$ denotes the reconstruction operator implemented by DeepOp-Informed. We show that DeepOp-Informed provides improved source estimates compared with both physics-agnostic and classical reconstruction methods.}
\label{fig:intro}
\end{figure}%

DeepOp-Informed is also related to traditional physics-informed approaches \parencite[see, e.g.,][]{azzimonti2015blood, menicali2026physicsinformed, kashinath2021physicsinformed, raissi2019physicsinformed, subramanian2022adaptive, zhang2023filtered}. However, such methods typically encode physical information by augmenting the loss function with physics-derived penalties, rather than incorporating the biophysical process underlying the sensing mechanism. Moreover, unlike current deep learning approaches to operator learning \parencite{goswami2023physicsinformed, kovachki2024neural, lu2021learning}, DeepOp-Informed can accommodate sample-specific variations in the physics of signal generation.

We demonstrated the effectiveness of DeepOp-Informed for MEG source reconstruction from a single time snapshot in both realistic simulation scenarios and MEG recordings from a cohort of adolescents. On realistic simulated data, where ground truth is known (Figure~\ref{fig:simulation}), DeepOp-Informed achieved a threefold reduction in reconstruction error relative to all other methods considered. Moreover, the application to MEG data from adolescents (Figure~\ref{fig:application}) shows that DeepOp-Informed can learn from as few as 200 training samples, yielding physiologically plausible source reconstructions and generalizing across subjects.

Although we demonstrated the method in the context of MEG, the framework is general and can be applied to a broad range of inverse or reconstruction problems.


\begin{figure}[!htb]%
\centering
\includegraphics[width = 1\textwidth]{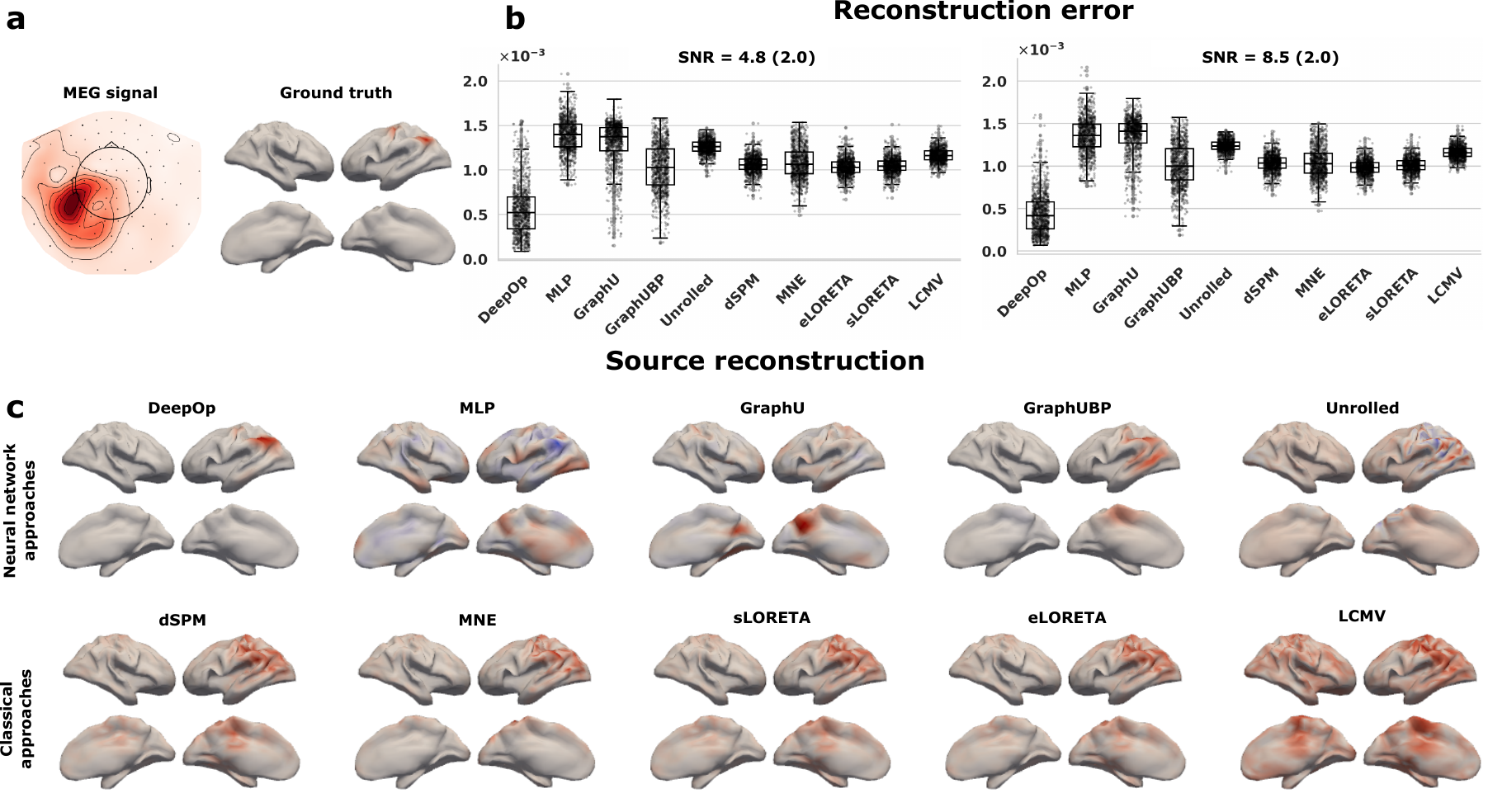}
\caption{\textit{MEG reconstruction performance of DeepOp-Informed compared with existing approaches.} \textbf{a.} Example of a synthetic test MEG sensor signal $y_i$ and its associated ground-truth source signal $x_i$. Training synthetic source signals $x_i \in \mathbb{R}^p$, with $p = 1284$, were generated on the reference cortical surface as one to three equal-intensity coactive peaks, with the number of peaks sampled uniformly from $\{1,2,3\}$ for each observation. Next, a forward model $\mathcal K_i \in \mathbb{R}^{204 \times 1284}$ was selected at random from a set of 13 subjects, and the sensor signal was generated as $y_i = \mathcal K_i x_i + \epsilon_i$, where $\epsilon_i$ denotes random noise. The task was to reconstruct the source signals $x_i$ from previously unseen noise-contaminated MEG measurements $y_i$. For methods that require training, 200 such observations were used as the training set. In addition, 1,000 further test observations were generated to evaluate performance, using forward models from six held-out subjects. For models requiring estimation of the source or noise covariance, these quantities were estimated from a separate validation set of 1,000 observations. \textbf{b.} Reconstruction error, measured as $\left\| x_i/\|x_i\|_2 - \hat x_i/\|\hat x_i\|_2 \right\|_2^2$, where $\hat{x}_i$ denotes the reconstructed source signal, was evaluated across 1,000 test observations for DeepOp-Informed and nine baseline approaches (see Section~\ref{sec:simulations} for a description). To standardize comparisons across methods, both $x_i$ and $\hat x_i$ were normalized to have Euclidean norm 1 before the reconstruction error was computed. \textbf{c.} Example reconstructions of test signals obtained with DeepOp-Informed (labeled DeepOp in the figure), together with alternative neural-network-based and classical baseline methods. DeepOp-Informed substantially improves reconstruction relative to the baseline methods.
}
\label{fig:simulation}
\end{figure}%

\section{Results}\label{sec:results}
\subsection{Overview of DeepOp-Informed}\label{sec:overview}
We assume that the data are generated according to the model
\begin{equation}
y_i = \mathcal K_i x_i + \epsilon_i, \qquad i=1,\ldots,n,
\end{equation}
where $x_i \in \mathbb{R}^p$ is the source signal, $y_i \in \mathbb{R}^s$ is the corresponding sensor measurement, $\mathcal K_i \in \mathbb{R}^{s \times p}$ is a potentially subject-specific forward model, and $\epsilon_i$ denotes observational noise. DeepOp-Informed learns a neural-network-based reconstruction operator through its parameters $\theta$,
\[
\mathcal{R}_{\theta,\mathcal K}:\mathbb{R}^s \to \mathbb{R}^p,
\]
which maps sensor-space measurements $y_i \in \mathbb{R}^s$ to source-space estimates $x_i \in \mathbb{R}^p$ for a given user-specified forward model $\mathcal K$. The architecture consists of two components. First, the differentiable biophysics-informed layer produces a preliminary reconstruction by combining the forward model $\mathcal K_i$ with a learned regularization. Under a working joint Gaussian assumption on the signal $x_i$ and unit-white noise $\epsilon_i$, this layer can be interpreted as computing
\begin{align}
\mathcal B_{\theta_1,\mathcal K_i}(y_i)
=&
\mathbb{E}\!\left[x_i \mid y_i, \mathcal K_i\right]\\
=& \left( {\mathcal K}_i^T {\mathcal K}_i + \mathcal{P}_{\theta_1} \right)^{-1} {\mathcal K}_i^T y_i.
\end{align}
The prior covariance structure of $x_i$ enters this conditional expectation through its inverse, $\mathcal P_{\theta_1}\in\mathbb{R}^{p\times p}$, which is parameterized by $\theta_1$ as a linear combination of powers of the Laplacian and learned end-to-end. This layer can be regarded as a generalization of a convolutional layer: it is linear, incorporates spatial information, is parameterized by a low-dimensional parameter vector, $\theta_1$, and is differentiable. Unlike a convolutional layer, however, it defines a global kernel \parencite{silverman1984spline} that disentangles the mixing introduced by the forward model. Several such layers, each with its own learnable parameters, can be combined in a manner similar to convolutional layers. This initial estimate $\mathcal B_{\theta_1,\mathcal K_i}(y_i)$, located on the source space, is refined by a geometric deep-learning module,
\[
\mathcal F_{\theta_2}:\mathbb{R}^p \to \mathbb{R}^p,
\]
implemented here as a graph neural network on the cortical mesh. The full reconstruction model is thus given by the composition
\[
\mathcal R_{\theta,\mathcal K_i}
=
\mathcal F_{\theta_2}\circ \mathcal B_{\theta_1,\mathcal K_i},
\qquad
\theta=(\theta_1,\theta_2).
\]
DeepOp-Informed is then trained to minimize the mean squared error (MSE)
\[
\frac{1}{n}\sum_{i=1}^n \left\| x_i - \mathcal R_{\theta,\mathcal K_i}(y_i) \right\|^2
\]
on synthetic data $(y_i, x_i, \mathcal K_i)$. This design embeds mechanistic information from the forward model directly into the architecture, rather than requiring the network to infer it solely from training pairs $(y_i, x_i)$. DeepOp-Informed can be applied across subjects with different subject-specific forward models, accommodating variations in the number of channels due to the exclusion of bad channels, as well as differences in brain and head geometry.

\subsection{Simulation experiments}\label{sec:simulations}
\subsubsection{Synthetic data generation}\label{sec:data_generation}
Synthetic data were generated both to validate the performance of the proposed model and to train the model for subsequent application to real MEG source reconstruction. Specifically, we used a triangular mesh representing a reference cortical surface, consisting of $p=1284$ vertices and 2,560 triangles, as the source space (Figure~\ref{fig:simulation}a). Subject-specific forward models ${\mathcal K}_i$ were computed from a cohort of 19 adolescents recruited at the University of Washington Institute for Learning \& Brain Sciences. Using MNE \parencite{gramfort2014mne}, each forward model was derived from the subject’s MRI and represented as a $s \times p$ matrix, where $s = 204$ denotes the number of gradiometer measurements acquired by the MEG system used in the study. All forward models were rescaled to have Frobenius norm 1, in order to mitigate floating-point approximation errors during training.

We generated synthetic data on the cortical surface as follows: source signals $x_i$ were created by generating a random number of equal-intensity localized sources centered at randomly selected vertices on the cortical surface, followed by smoothing to produce spatially localized but more diffuse peaks (see Figure~\ref{fig:simulation}a, Ground truth, for an example source signal). The training and test data were generated using non-overlapping sets of source vertices, each comprising 50\% of the cortical surface vertices, from which the point sources were sampled. These source signals represent the situation in which one or more brain regions are coactive during a specific task. We then generated $n_{\text{train}} = 200$ triplets $(y_i, x_i, \mathcal{K}_{\pi(i)})$ for training and $n_{\text{test}} = 1{,}000$ triplets for testing, where
\begin{equation}\label{eq:multi_subj_sim}
y_i = \mathcal{K}_{\pi(i)} x_i + \epsilon_i,
\end{equation}
with $\pi(i)$ sampled from $\{1,\ldots,13\}$ for the training set and from $\{14,\ldots,19\}$ for the test set, meaning training and test data use different subsets of the 19 forward models available. The noise term $\epsilon_i$ was modeled as a Gaussian random vector with mean zero and covariance matrix $\Sigma^\epsilon_{\pi(i)}$, which was computed from real MEG recordings of subject $\pi(i)$ using time windows preceding stimulus onset in auditory evoked-response experiments. This covariance matrix was normalized and multiplied by a factor of $\sigma^2$ to control the SNR, with $\sigma = 0.5$ or $\sigma = 1.2$, corresponding to average SNRs of 8.5 and 4.8, respectively, with a standard deviation of 2.  

\subsubsection{Reconstruction methods}\label{sec:approaches}
We compared the following models:
\begin{itemize}
\item[] \textbf{MLP}: A single hidden-layer perceptron with 204 input variables, 1,284 hidden-layer variables, and 1,284 output variables. A tanh nonlinear activation function was used.

\item[] \textbf{GraphU}: A model based on the one proposed in \textcite{zhu2018image}, adapted to the setting of signals supported on a nonlinear surface. It consists of an MLP, as described above, followed by a Graph U-Net \parencite{gao2022graph} with 32 hidden channels, 3 pooling layers, and a node pooling rate of 0.5 at each layer. The graph nodes were defined by the 1,284 vertices of the reference mesh, and the edges were derived from the triangles connecting these vertices.

\item[] \textbf{GraphUBP}: A model in which the fully connected module of GraphU was replaced by back-projected data, using $\mathcal{K}_i^T y_i$ as input. This represents the simplest approach for mapping data from sensor space to source space, after which the \textbf{GraphU} module was applied as described above. Replacing $\mathcal{K}_i^T y_i$ with the pseudoinverse projection $\mathcal{K}_i^\dagger y_i$ led to a further deterioration in performance.

\item[] \textbf{Unrolled}: An unrolled neural-network based on the architecture of \cite{aggarwal2019modl}, with 10 unrolled iterations, in which each stage consists of a data-consistency update followed by a learned denoising step implemented using the \textbf{GraphU} architecture described above. Although not specifically designed for this purpose, the approach allows the incorporation of subject-specific forward-model information.

\item[] \textbf{DeepOp-Informed}: Our proposed reconstruction model, $\mathcal R_{\theta,\mathcal K_i}$. The inverse source covariance in the biophysics-informed input layer was parameterized as the sparse matrix
\[
\mathcal{P}_{\theta_1} = \theta_1^{(0)} I + \theta_1^{(1)} \Delta, \qquad
\theta_1=(\theta_1^{(0)},\theta_1^{(1)}),
\]
where $I \in \mathbb{R}^{p \times p}$ is the identity matrix and $\Delta \in \mathbb{R}^{p \times p}$ is the discretized Laplace--Beltrami operator \parencite{dziuk1988finite} on the reference cortical mesh (see Section~\ref{sec:methods} for details). Parameterizing the inverse covariance $\mathcal{P}_{\theta_1}$ in this way allows the input layer to incorporate the spatial structure of the cortical surface into the preliminary reconstruction. The refinement module $\mathcal F_{\theta_2}$ was implemented using the same Graph U-Net architecture as in \textbf{GraphU}, with 32 hidden channels.

\item[] \textbf{Classical reconstruction methods.} Several classical linear reconstruction methods were also considered as baselines, including dSPM \parencite{dale2000dynamic}, MNE \parencite{hamalainen1994interpreting}, sLORETA \parencite{pascual-marqui2002standardized}, eLORETA \parencite{pascual-marqui2007discrete}, and LCMV \parencite{vanveen1997localization}, as implemented in MNE \parencite{gramfort2014mne}. As these are standard approaches, we do not describe them in detail here and instead refer the reader to the original references for additional information.
\end{itemize}

All neural-network approaches were trained on the same synthetic data. The matrices $\mathcal{K}_{\pi(i)}$ were provided to DeepOp-Informed, GraphUBP, the unrolled model, and the classical reconstruction methods. For classical linear reconstruction methods that require estimation of the source or noise covariance, these quantities were estimated from 1,000 observations, and the corresponding hyperparameters were tuned on the same set via cross-validation. The MLP, GraphU, GraphUBP, Unrolled, and DeepOp-Informed were trained using the MSE loss and the Adam optimizer \parencite{kingma2017adam}, for 300, 100, 100, 40, and 40 epochs, respectively. To limit computation time, the other methods were trained for fewer epochs than the MLP baseline. A learning rate of 0.01 was used for all methods.

\subsubsection{Reconstruction performance}

We ran 20 Monte Carlo repetitions. In each repetition, $n_{\text{train}} = 200$ triplets were used for training and $n_{\text{test}} = 1{,}000$ triplets for testing. The simulations were repeated under varying numbers of coactive source peaks, with the number of peaks per observation drawn at random from the sets $\{1\}$, $\{1,2\}$, or $\{1,2,3\}$. Simulations were performed under two noise settings: $\sigma = 0.5$ or $\sigma = 1.2$. 

For a single simulation repetition with $\sigma = 0.5$, in which the number of coactive sources was drawn uniformly at random from $\{1,2,3\}$, we compared reconstruction performance across 1,000 test observations using the squared error $\left\| x_i/\|x_i\|_2 - \hat x_i/\|\hat x_i\|_2 \right\|_2^2$, where $\hat{x}_i$ denotes the reconstructed source signal (Figure~\ref{fig:simulation}b). For the example sensor signal $y_i$ and its associated ground-truth source signal $x_i$ shown in Figure~\ref{fig:simulation}a, the corresponding reconstructions obtained with the different methods are shown in Figure~\ref{fig:simulation}c.

DeepOp-Informed outperformed all other models, demonstrating its ability to effectively leverage the additional physics information encoded in the forward model $\mathcal{K}_i$, despite the ill-posed nature of the inverse problem.
The next best-performing methods were GraphUBP and the classical reconstruction approaches, which are also able to incorporate information from the subject-specific forward model. The performance of the unrolled model appears to be more strongly affected by the ill-posedness of the MEG inverse problem, which is more severe than that encountered in MRI, the setting for which such architectures were originally developed. That said, we cannot rule out the possibility that alternative adaptations of this architecture to the MEG setting could achieve better results. By contrast, MLP and GraphU performed worst in this setting, likely because they do not have direct access to the subject-specific forward model and therefore struggle to generalize to unseen subjects. Next, for each simulation repetition, we computed the average MSE,
$n_{\text{test}}^{-1}\sum_{i=1}^{n_{\text{test}}}\left\|x_i/\|x_i\|_2-\hat{x}_i/\|\hat{x}_i\|_2\right\|_2^2$, and compared these averages across the 20 repetitions within each simulation setting (Supplementary Figure~\ref{fig:realistic_subject_specific}). These results indicate that the observed threefold reduction in reconstruction error is robust to differences in the number of coactive regions and to random variation in the training data.

Reconstruction performance was also assessed using the receiver operating characteristic curve (AUC). Specifically, to compute the AUC, the true source signal was normalized to have a maximum value of 1, and source locations with values greater than 0.9 were labeled as active. The corresponding reconstructed values were then used as prediction scores. DeepOp achieved a significantly higher mean AUC than the next-best-performing method in a two-sided paired \(t\)-tests across Monte Carlo repetitions (Supplementary Figure~\ref{fig:realistic_subject_specific_AUC}), indicating that its advantage extended to identifying active source regions and was not limited to normalized reconstruction error.

When DeepOp-Informed was provided with an incorrect forward model, it produced a noticeably different reconstruction (Supplementary Figure~\ref{fig:Multi-subj-wrong}). This finding highlights that DeepOp-Informed actively uses forward-model information to produce subject-specific source estimates.

We also repeated the experiment focusing on a single study participant; that is, we generated data according to
\begin{equation}\label{eq:single_subj_sim}
y_i = \mathcal{K} x_i + \epsilon_i,
\end{equation}
where $\mathcal{K}$ is the (fixed) forward model for the chosen subject. The noise term $\epsilon_i$ was modeled as a Gaussian random vector with mean zero and covariance matrix $\Sigma^\epsilon$, computed from real MEG recordings of that subject using time windows preceding stimulus onset in evoked-response experiments. The subject's own cortical surface mesh was used both to construct the graph for GraphU and to discretize the Laplace--Beltrami operator used in DeepOp-Informed. In this setting, DeepOp-Informed still outperformed the other methods, although the reconstruction performance of both MLP and GraphU improved (Supplementary Figure~\ref{fig:realistic_single_subject}). This is expected, because using the same forward model for both training and testing makes generalization easier for models that do not explicitly incorporate forward-model information.

\begin{figure}[!htb]%
\centering
\includegraphics[width = 0.8\textwidth]{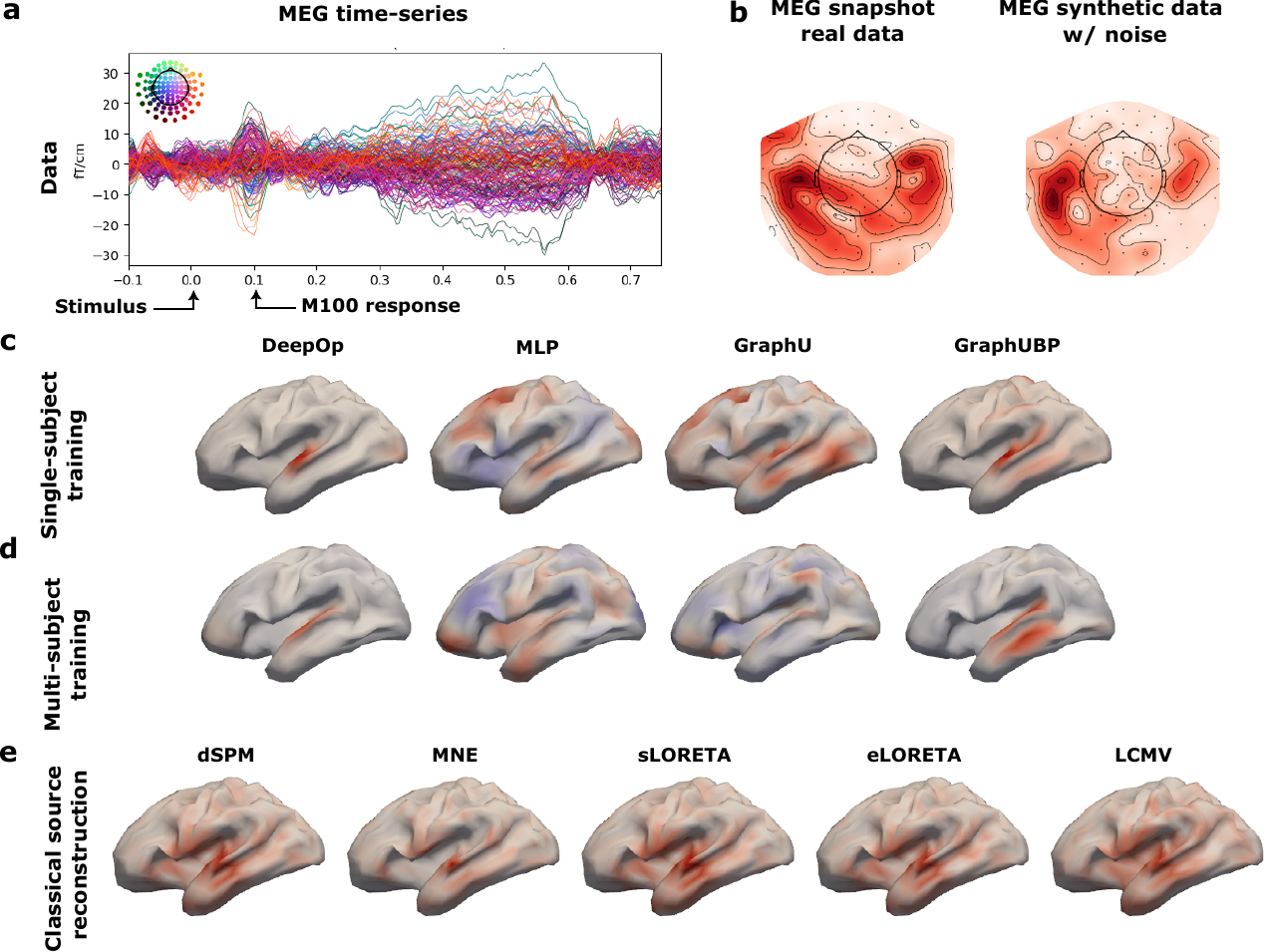}
\caption{\textit{Source reconstruction of MEG recordings.} \textbf{a.} MEG evoked-response time series obtained by averaging 420 trials corresponding to the first syllable of a speech stream consisting of trisyllabic pseudowords (details in Section~\ref{sec:application}). The timing of auditory stimulus onset and the M100 response, which occurs approximately 100 ms after stimulus onset, are annotated. \textbf{b.} MEG sensor signal associated with the identified M100 evoked response, shown with an example of a synthetic MEG sensor signal used for training. \textbf{c.} Source reconstruction of the M100 evoked response using neural-network-based approaches, including the proposed DeepOp-Informed (DeepOp in the figure), in the single-subject training setting (details in Section~\ref{sec:application}). \textbf{d.} Source reconstruction of the M100 evoked response using neural-network-based approaches in the multi-subject training setting (details in Section~\ref{sec:application}) for the same subject used in the single-subject training setting, who was held out during training. \textbf{e.} Source reconstruction using classical reconstruction methods. Refer to Section~\ref{sec:approaches} for a description of all compared methods. Although ground truth is unknown in this setting, DeepOp-Informed and GraphUBP yielded more plausible reconstructions, with DeepOp-Informed producing activation patterns that were more sharply localized to the auditory cortex, as expected in an auditory experiment. Classical source reconstruction methods identified regions that included the auditory cortex, but the estimated activation was spread over broader areas, appearing less aligned with cortical functional organization and more reflective of physical proximity to sensors with greater signal amplitude. Results for the unrolled neural network departed more markedly from expected patterns (Supplementary Figure~\ref{fig:Application_unrolled}), suggesting that further architectural modifications may be needed to adapt it to the MEG setting.}
\label{fig:application}
\end{figure}%

\subsection{Illustration on MEG recordings}\label{sec:application}
We applied DeepOp-Informed, together with the alternative approaches described in Section~\ref{sec:approaches}, to MEG data from the cohort of adolescent subjects used in the simulation study to obtain the forward models and prestimulus covariance structures. Our analysis focused on the signals recorded by the 204 gradiometer channels of the MEG system used in the study.

The models were applied to MEG recordings from auditory evoked-response experiments (Figure~\ref{fig:application}a). These experiments involved the presentation of speech streams composed of trisyllabic pseudowords, that is, sound sequences that resemble typical English phonetic and orthographic patterns but lack semantic meaning. The speech stream was designed such that all pseudowords occurred with equal frequency and the transitional probabilities between words were uniform. Specifically, the first syllable (S1) of each pseudoword perfectly predicted the subsequent two syllables (S2 and S3), whereas the final syllable provided no information about the following word, eliminating cross-word cues. Participants also received visual feedback; however, this component was not included in our analysis.

For each syllable position within the pseudowords, evoked responses were obtained by averaging the MEG signals across 420 trials. The evoked response associated with the first syllable is shown in Figure~\ref{fig:application}a,b. The M100 evoked response, which occurs approximately 100 ms after the onset of the auditory stimulus, was identified for each of the three evoked responses, and the corresponding source signal was reconstructed. The neural-network models used were those defined and trained as described in the simulation experiments in Section~\ref{sec:simulations}, with synthetic training data generated using noise level $\sigma = 1.2$ and the number of coactive sources sampled uniformly from $\{1,2,3\}$. Here, we focused on the left-hemisphere M100 response, which, based on the sensor-level signal, appears to provide a relatively clean auditory evoked response for evaluating spatial localization in a setting where the expected location of neural activation is reasonably well understood.

\subsubsection{Multi-subject training}
The models trained on synthetic data generated from the forward models and prestimulus covariance structures of 13 adolescents, as described in equation~\eqref{eq:multi_subj_sim}, were then applied to reconstruct the source activity underlying the M100 evoked response to the first syllable (stimulus condition S1) in a held-out subject.

Figure~\ref{fig:application}d shows that, although the ground truth is unknown in this setting, MLP and GraphU produced reconstructions that deviated substantially from the spatial patterns typically expected in auditory evoked-response experiments. The unrolled neural network departed even more markedly from these expected patterns (Supplementary Figure~\ref{fig:Application_unrolled}). DeepOp-Informed and GraphUBP yielded more plausible reconstructions, with DeepOp-Informed producing activation patterns that were more sharply localized to the auditory cortex. These differences likely reflect, at least in part, whether the methods explicitly incorporate information from the forward model, which in turn affects their ability to adapt to the unavoidable mismatch between the synthetic MEG data used for training and the real MEG data acquired from human subjects. The source reconstructions of the evoked responses to the second and third syllables in the auditory experiment, which can be viewed as perturbations \parencite{yu2020veridical} of the evoked response to the first syllable, show that DeepOp-Informed provides robust reconstructions while also producing sharp and physiologically plausible source estimates (Supplementary Figure~\ref{fig:Application_sens}). Classical source reconstruction methods were also applied (Figure~\ref{fig:application}e). These methods identified regions that included the auditory cortex, but the estimated activation extended over larger areas that appeared less aligned with cortical functional units and more reflective of physical proximity to sensors with greater signal amplitude.

\subsubsection{Single-subject training}
Additionally, we evaluated a single-subject model trained on synthetic data generated according to equation~\eqref{eq:single_subj_sim}, using the subject’s own forward model, cortical surface, and prestimulus covariance structure. The resulting reconstructions closely resembled those of the multi-subject model, but with even sharper localization of activation, particularly for DeepOp-Informed and, to a lesser extent, GraphUBP. This increased sharpness may be explained by the absence of uncertainty arising from generalization to unseen forward models and the absence of cortical registration error, as the single-subject model relies on the subject’s cortical surface rather than on a common reference surface. This class of models, however, requires separate training for each subject, making it less attractive for broader clinical use.

\subsection{Additional experiments}\label{sec:additional_simulations}

We conducted additional experiments to evaluate the proposed model in a more controlled setting, where the forward models could be generated explicitly. These experiments were designed specifically to compare DeepOp-Informed with physics-agnostic neural-network architectures and to assess the benefit of incorporating physics-based information into neural-network models. For these additional experiments with a smaller source space, we ran 50 Monte Carlo repetitions for each setting.

\noindent \textbf{Spherical source setting.} 
We used an icosahedral sphere with 642 vertices as the source space. We then generated 128 points equidistant from the sphere, as shown in Supplementary Figure~\ref{fig:sphere-simulations}, to simulate sensor positions. Next, we generated 200 random permutations of these 128 sensor locations and computed the corresponding forward models, $\{\mathcal{K}_j \in \mathbb{R}^{128 \times 642}: j = 1,\ldots, N=200\}$, using a realistic MEG forward model \parencite{oostenveld2011fieldtrip}. The 642 vertices were treated as dipole locations, with dipole orientations fixed orthogonally to the sphere, and the 128 points were used as sensor positions. Synthetic data were then generated according to equation (\ref{eq:multi_subj_sim}), yielding $n_{\text{train}}$ triplets $(y_i, x_i, \mathcal{K}_{\pi(i)})$ for training and $n_{\text{test}}$ triplets for testing, where $\pi(i)$ was sampled from $\{1,\ldots,100\}$ for the training set and from $\{101,\ldots,200\}$ for the test set. DeepOp-Informed was informed by the $\mathcal{K}_{\pi(i)}$ component of the triplet, which was also used to generate back-projected data for the GraphUBP approach. The simulations were repeated with a single active source or random counts of coactive sources per observation drawn uniformly from the sets $\{1,2\}$, or $\{1,2,3\}$. Simulations were performed under two settings, with noise covariance matrix $\Sigma_\epsilon = \sigma^2 I$ and $\sigma = 0.05$ or $\sigma = 0.2$, respectively. We used $n_{\text{train}} = 200$ training observations and $n_{\text{test}} = 1{,}000$ test observations. 

We compared the performance of DeepOp-Informed in the multi-subject training setting with that of models trained under the same simulation setup but with a fixed forward model. DeepOp-Informed matched the performance of its single-subject counterpart and still outperformed MLP, GraphU, and GraphUBP, despite the latter models having been trained in the simpler fixed-forward-model setting (Supplementary Figure~\ref{fig:sphere-subject-specific}). Interestingly, GraphUBP, which performed favorably in the realistic MEG setting, performed worse in this case, likely because of the even lower SNR. 

In the fixed-forward-model setting, the performance of MLP and GraphU improved with increasing sample size. However, MLP did not match the performance of DeepOp-Informed even with 20,000 training samples. GraphU required more than 5,000 training samples in the low-SNR setting and more than 10,000 in the high-SNR setting to reach the performance achieved by DeepOp-Informed with only 200 training samples (Supplementary Figure~\ref{fig:sphere-sample-size}).

\noindent \textbf{Robustness to forward model misspecification.} We evaluated the robustness of DeepOp-Informed to misspecification of the forward model. Data were generated in the spherical-source setting using a fixed forward model $\mathcal{K}$, while reconstruction was carried out using a perturbed surrogate model $\tilde{\mathcal{K}}$ instead of $\mathcal{K}$. The perturbation was applied independently to each entry of $\mathcal{K}$ and sampled from a uniform distribution, $\mathrm{Unif}(0, c\, k_{\max})$, where $k_{\max}$ is the largest entry of $\mathcal{K}$ and $c \in \{0, 0.01, 0.025, 0.05\}$. DeepOp-Informed continued to benefit from forward-model information and outperformed the baseline neural network models over a reasonable range of perturbation magnitudes. In lower-SNR regimes, it maintained improved performance even under larger perturbations (Supplementary Figure~\ref{fig:sphere-perturb}). In higher-SNR regimes, performance deteriorated more rapidly as the perturbation magnitude increased.

\noindent \textbf{Multi-head extension.} We evaluated a multi-head extension of DeepOp-Informed in the fixed-forward-model simulation scenario. In this architecture, multiple biophysics-informed layers are stacked in parallel, each with its own set of learnable parameters. Their outputs are concatenated and passed to a graph neural network, such that each node of the graph receives a feature vector with dimensionality equal to the number of heads, as described in Section~\ref{sec:extensions}. Specifically, we assessed the reconstruction performance of DeepOp-Informed with two-head and three-head biophysics-informed layers. The multi-head architecture improved performance, particularly in higher-SNR settings (Supplementary Figure~\ref{fig:sphere-multi-head}).

\noindent \textbf{Direct measurements.} We also evaluated DeepOp-Informed in a special case of the sample-specific forward-model setting in which the signal is measured directly in source space, but only at randomly sampled locations. This setting can be viewed as loosely analogous to invasive modalities such as intracranial EEG, with different sensor placement across subjects. Specifically, we generated triplets $(y_i, x_i, \mathcal{K}_i)$, with source signals $x_i$ generated as in the spherical source setting. To represent direct sensing at random source-space locations, each forward model $\mathcal{K}_i$ was constructed as an $s \times 642$ matrix with exactly one nonzero entry per row, located at a randomly selected column, and zeros elsewhere. We considered simulations with $s \in \{10, 50, 100\}$. We then compared reconstruction performance across all models considered. DeepOp-Informed outperformed all baseline methods (Supplementary Figure~\ref{fig:sphere-missing-data}).

\section{Discussion}\label{sec:discussion}
In this paper, we presented DeepOp-Informed, a novel framework that enables the regularization of neural network models by embedding information about the underlying physics of an ill-posed inverse problem. We showed that the proposed model leads to more stable and efficient learning. This advantage is particularly relevant in settings where data are scarce, making physics-agnostic neural network approaches infeasible due to their large number of parameters. In the context of electrophysiological source reconstruction, where training data are synthetically generated, this approach enables the model to learn to reconstruct the underlying source signals from a smaller curated dataset, without requiring extensive training data that attempt to account for all possible discrepancies between synthetic and real data, which is especially challenging in non-normative populations. The proposed framework can also be understood as a simulation-based approach for injecting prior information about the source signal into the neural-network model.

In contrast to other biophysics-informed approaches \parencite{sun2022deep}, DeepOp-Informed explicitly incorporates physical information through a novel custom layer, enabling the model to accommodate variations in the underlying physics across samples. In the context of electrophysiological source reconstruction, these properties enable DeepOp-Informed to account for subject-specific differences in the forward model, allowing it to be trained once and then applied to new subjects without retraining or expert intervention. The proposed application also leverages synthetically generated data, introducing an additional degree of freedom for modeling acquisition non-idealities that are otherwise difficult to incorporate directly into the model, such as realistic noise and prior knowledge about the localized nature of source signals. Future work could incorporate head-movement effects into synthetic data to improve model robustness to movement-related artifacts and compare performance with established signal-space separation preprocessing pipelines \parencite{taulu2005presentation, taulu2006spatiotemporal} commonly used in MEG preprocessing, which explicitly compensate for head-movement-induced distortions. 

Several other extensions are possible, including incorporating subject-specific cortical structure into the multi-subject version of the model so that subject-specific geometry can be leveraged more explicitly, as well as quantifying uncertainty of reconstructions and incorporating uncertainty in the forward model itself. Another important direction is to validate DeepOp-Informed’s ability to learn from training data partially composed of phantom data. In this setting, DeepOp-Informed could provide a practical mechanism for transferring information from phantoms, where the ground-truth activation is known, to MEG source reconstruction in human subjects.

The proposed contribution can also be situated within the literature on physics-informed neural networks \parencite[see, e.g.,][]{kashinath2021physicsinformed, raissi2019physicsinformed, subramanian2022adaptive, zhang2023filtered}, which typically encode physical constraints by augmenting the loss function with physics-derived penalty terms. However, such methods are generally designed to solve a single instance of a ``reconstruction'' problem. DeepOp-Informed is also related to deep operator learning \parencites{goswami2023physicsinformed, kovachki2024neural, lu2021learning}, which aims to learn mappings between function spaces from training data. In contrast to this class of approaches, DeepOp-Informed can accommodate sample-specific variations in the physics of signal generation.

In summary, DeepOp-Informed introduces a new approach for source reconstruction, image reconstruction, and potentially other tasks involving sensor-level data, since these can be addressed by replacing the refinement module with a task-specific network. The framework provides a new way to incorporate physics-based models into deep-learning architectures, improving data efficiency by leveraging mechanistic information while still allowing the data to guide the learning process. Although demonstrated here in the context of MEG source reconstruction, we expect this paradigm to extend naturally to a broad range of imaging modalities, including MRI, computed tomography, and positron emission tomography.

\section{Methods}\label{sec:methods}

In its most general form, and in line with formulations used in the inverse-problems literature \parencite{arridge2019solving} and the operator-learning literature \parencite{goswami2023physicsinformed, kovachki2024neural, lu2021learning}, the proposed model can be described as follows. Let $\Omega_x$ denote the spatial domain of the source signal, and let $\mathcal X$ denote a function space of scalar signals $x:\Omega_x \to \mathbb R$. Furthermore, let $\mathcal K \in \mathcal L(\mathcal X;\mathcal Y)$ denote the (potentially random) forward operator, mapping a signal $x \in \mathcal X$ to an observation $\mathcal K(x) \in \mathcal Y$. In its more general form, the objective is to recover the \textit{source signal} $x \in \mathcal X$ from a noisy \textit{sensor signal} $y \in \mathcal Y$, assumed to follow
\begin{equation}\label{eq:inv_model_intro}
y = \mathcal K(x) + \epsilon,
\end{equation}
where $\epsilon \in \mathcal Y$ represents observational noise. 

Here, however, we focus on the MEG source-reconstruction setting and assume that the spatial domain of the source signal, $\Omega_x$, is a triangular mesh $\mathcal M$ embedded in $\mathbb R^3$, with $p$ vertices representing the cortical surface. The source signal is modeled as a vector of values at the mesh vertices, $x \in \mathbb R^p$. Hence, the mesh vertices are also used as dipoles, oriented perpendicularly to the mesh. The vector $x$ can also be viewed as defining a piecewise linear function on the mesh, obtained by linear interpolation within each triangle, thereby representing a continuum of dipoles on the cortical surface. The sensor signal is represented by $y \in \mathbb R^s$, corresponding to measurements from an array of $s$ MEG sensors that provide a spatially discretized representation of the extracranial magnetic field. The forward model is assumed to be linear and is represented by a matrix $\mathcal K \in \mathbb R^{s \times p}$, which describes how the source signal $x$ on the cortical surface gives rise to the measured magnetic field, accounting for the physical properties of biological tissues and head geometry.

Formally, the task is to define a reconstruction operator $\mathcal{R}_\theta$, which, in the MEG setting considered here, reduces to a map $\mathcal{R}_\theta: \mathbb{R}^s \to \mathbb{R}^p$, where the finite-dimensional vector $\theta \in \Theta$ provides a parameterization of this map. The family $\{ \mathcal{R}_\theta : \theta \in \Theta \}$ defines a class of candidate reconstruction operators, where the ``best'' reconstruction operator is denoted by $\mathcal{R}_{\hat \theta}: \R^s \to \R^p$. In learning-based approaches, an optimal parameter vector $\hat{\theta} \in \Theta$ is obtained by minimizing the empirical risk:

\begin{equation}\label{eq:loss}
\hat{\theta} \in \arg \min_{\theta \in \Theta} \frac{1}{n} \sum_{i=1}^{n} \ell \left(x_i, \mathcal R_{\theta}(y_i)\right),
\end{equation}
where $(y_i, x_i) \in \R^s \times \R^p$ are i.i.d. samples of $(y, x)$ from equation (\ref{eq:inv_model_intro}). The function $\ell: \R^p \times \R^p \to \mathbb{R}$ denotes the loss function, which in this work is taken to be the squared Euclidean distance $\ell(x,x') = \|x - x' \|_2^2$.

Fully learned approaches model $\mathcal{R}_\theta$ with a neural network of generic architecture, ignoring the structure encoded in the forward model and instead fully relying on the training examples $(y_i, x_i)$ to learn the mapping from raw measurements to source data. To account for the non-local mixing introduced by the forward model, fully learned approaches typically include one or more \textit{fully connected layers} that serve as a pseudoinverse map $\mathcal B_{\theta_1}: \R^s \to \R^p$, mapping elements from $\R^s$ to $\R^p$. This is followed by a \textit{post-processing/refinement} neural network $\mathcal F_{\theta_2} : \R^p \to \R^p$ \parencite{arridge2019solving}, leading to the reconstruction model:
\[
\mathcal R_{\theta} := \mathcal F_{\theta_2} \circ \mathcal B_{\theta_1},
\]
where $\theta = (\theta_1, \theta_2)$ and $\circ$ denotes composition. An example of this approach, although in the context of MRI data, is found in \cite{zhu2018image}, where $\mathcal B_{\theta_1}: \mathbb{R}^s \ra \mathbb{R}^p$ is implemented as a feedforward neural network with two fully connected layers, and $\mathcal F_{\theta_2}: \mathbb{R}^p \ra \mathbb{R}^p$ is a convolutional auto-encoder. 

A key drawback of this class of neural networks is the large number of parameters introduced by the fully connected layers, which requires a large training set to learn these parameters. A straightforward modification is to restrict learning to $\mathcal{F}_{\theta_2}$ (i.e., the post-processing stage) and use, for instance, the pseudoinverse to define the map $\mathcal{B}(y) = \mathcal{K}^\dagger y$. However, this initial reconstruction step may not be optimal for subsequent learning-based refinements, leading to suboptimal performance, as shown in our simulations.

As depicted in Figure~\ref{fig:intro}, and described in more detail in the next section, to address these limitations, we introduce a novel architecture for $\mathcal{B}$ that embeds information about the forward model while retaining a learnable component. Hence, we call this model biophysics-informed. The learnable parameters of $\mathcal{B}$ are optimized in an end-to-end fashion, that is, jointly with the rest of the architecture to provide a preliminary inversion that is designed to facilitate subsequent refinements. The proposed framework is not limited to source reconstruction and can be extended to other tasks --- such as classification --- by modifying the architecture of $\mathcal{F}_{\theta_2}$. The biophysics-informed layer $\mathcal{B}$ effectively resolves the mixing introduced by the forward model, enabling the use of ``local'' architectures, such as convolutional neural networks, for downstream tasks. These architectures are known to generally outperform fully connected networks by leveraging spatial structure.

\subsection{A novel biophysics-informed layer}
Let $\tilde{\mathcal K} \in \mathbb{R}^{s \times p}$ denote a surrogate forward model providing an approximation to $\mathcal K$. We define the \textit{biophysics-informed custom layer} $\mathcal{B}_{\theta_1, \tilde{\mathcal K}}$ as  
\begin{equation}\label{eq:inv_problem}
y \mapsto \mathcal{B}_{\theta_1, \tilde{\mathcal K}}(y) = \argmin_{x \in \R^p} J_{\theta_1, \tilde{\mathcal K}}(y, x),
\end{equation}
where 
\begin{equation}
J_{\theta_1, \tilde{\mathcal K}}(y, x) = \left\| y - \tilde{\mathcal{K}} x \right\|_2^2 + \left\| \mathcal{P}_{\theta_1}^{\frac{1}{2}} x \right\|^2_2.
\end{equation}
Here, $\mathcal{P}_{\theta_1} \in \R^{p \times p}$ is a sparse matrix parameterized by a vector $\theta_1$. In our application, we define $\mathcal{P}_{\theta_1}$ as  
\begin{equation}\label{eq:operator}
\mathcal{P}_{\theta_1} = \theta_1^{(0)} I + \sum_{m=1}^M \theta_1^{(m)} \Delta^m,
\end{equation}
where $I \in \mathbb{R}^{p \times p}$ is the identity matrix, $\Delta \in \mathbb{R}^{p \times p}$ is the finite-element discretization of the Laplace--Beltrami operator on $\mathcal M$ \parencites{dziuk1988finite, dziuk2013finite, lila2024interpretable}, $\Delta^m \in \mathbb{R}^{p \times p}$ denotes its $m$th power, and $\theta_1^{(m)}$ is the $m$th entry of $\theta_1$. We also assume that all entries of $\theta_1$ are nonnegative and that $\theta_1^{(0)}$ is strictly positive, so that 
$\mathcal{B}_{\theta_1, \tilde{\mathcal K}}$ is well defined.

The output $\mathcal{B}_{\theta_1, \tilde{\mathcal K}}(y)$ provides a preliminary reconstruction by combining information from the known forward model $\tilde{\mathcal{K}}$ and the prior encoded in the regularization term. The model in equation \eqref{eq:inv_problem}, with $\theta_1$ selected so that $\left\|\mathcal{P}^{\frac{1}{2}}_{\theta_1} x \right\|^2_2$ quantifies deviation from a prior that encodes domain knowledge about the reconstructed signal or image -- such as smoothness -- is widely used in the inverse problem literature \parencite{arridge2019solving} and in MEG source reconstruction \parencite{lin2006assessing}. When $\mathcal{P}_{\theta_1} = \theta_1^{(0)} I$, for a fixed $\theta_1^{(0)}$, the resulting operator corresponds to a Tikhonov-type inverse.

Under a working joint Gaussian assumption on the signal $x$ and unit-white noise $\epsilon$, $\mathcal B_{\theta_1,\tilde{\mathcal K}}$ can also be interpreted as computing
\begin{align}
\mathcal B_{\theta_1,\tilde{\mathcal K}}(y)
=&
\mathbb{E}\!\left[x \mid y, \tilde{\mathcal K}\right],
\end{align}
where the prior covariance structure of $x$ enters this conditional expectation through its inverse $\mathcal P_{\theta_1}  \in \R^{p \times p}$. Therefore, the proposed input layer can also be understood as an empirical Bayes approach.

Data-driven approaches to learning the penalty have also been explored in the inverse problems literature. The penalty may be parameterized explicitly \parencite{carlosdelosreyes2013image, delosreyes2017bilevel, kunisch2013bilevel, lunz2018adversarial} or implicitly defined through a learnable operator in an unrolled iterative minimization scheme \parencite[see, e.g.,][]{andrychowicz2016learning, gilton2019neumann, aggarwal2019modl}. However, these models are typically designed for settings with higher SNRs and more complex source-space structure, and they often involve deeper architectures with associated optimization and training-stability challenges.

In contrast, in this work, we propose using $\mathcal B_{\theta_1, \tilde{\mathcal K}}$ as a building block in a neural network architecture, where $\theta_1$ is learned from data. This is motivated by a key observation: due to the quadratic structure of $J_{\theta_1, \tilde{\mathcal K}}$, the mapping $y \mapsto \mathcal{B}_{\theta_1, \tilde{\mathcal K}}(y)$ in \eqref{eq:inv_problem} is linear with respect to $y$. Consequently, the map $\mathcal{B}_{\theta_1, \tilde{\mathcal K}}$ can be interpreted as a \textit{non-local} linear filter:  
\[
y \mapsto \mathcal B_{\theta_1, \tilde{\mathcal K}}(y) = \left(\tilde{\mathcal K}^T \tilde{\mathcal K} + \mathcal{P}_{\theta_1} \right)^{-1} \tilde{\mathcal K}^T y,
\]
where $\tilde{\mathcal K}^T$ denotes the transpose of $\tilde{\mathcal K}$. The map $\mathcal B_{\theta_1, \tilde{\mathcal K}}$ combines information from the known forward model $\tilde{\mathcal K}$ and the learned component $\mathcal{P}_{\theta_1}$ to reconstruct $x$ from $y$. Hence, the output of this filter can be interpreted as a preliminary inversion of the signal at the sensor level, which is then passed to the subsequent neural network module, $\mathcal{F}_{\theta_2}$, for further refinement. While this formulation shows that the proposed biophysics-informed layer is linear with respect to the sensor data $y$, explicitly computing this inverse matrix is not necessary, as we will demonstrate. As opposed to the model analyzed in \textcite{maier2019learning}, the forward model here can be severely ill-posed. Our model is also conceptually related to convolutional kernel networks \parencite{mairal2014convolutional}, where the kernel is implicitly determined by the forward model and the learnable penalty \parencite{silverman1984spline}.  

Multi-head extensions of the proposed biophysics-informed layer will also be introduced, enabling more complex preliminary inversions. Importantly, this preliminary inversion is designed to embed the physics information encoded in the forward model and to facilitate downstream processing through \textit{joint optimization} of $\theta_1$ and $\theta_2$ in an end-to-end manner. 

In the following, we describe the architecture of the refinement module $\mathcal{F}_{\theta_2}$, which in our application is designed to incorporate anatomical constraints of the brain surface where the source signal is located, to improve the reconstruction.

\subsection{Refinement module}
We introduce the following standard operations commonly used in geometric deep learning \parencite{bronstein2021geometric}:
\begin{itemize}
    \item A local linear layer $\mathcal A_{\theta}: \R^{p \times d_{\text{in}}} \to \R^{p \times d_\text{out}}$, where $d_{\text{in}}$ and $d_{\text{out}}$ denote the number of input and output features, respectively. An example, with $d_{\text{in}} = d_{\text{out}} = 1$, is $\mathcal A_{\theta} = \theta^{(0)} I + \sum_{m=1}^M \theta^{(m)} \Delta^m$, where $\Delta$ can be taken to be the graph Laplacian.
    The filter $\mathcal A_{\theta}$ provides a possible generalization of the concept of convolution to graphs \parencite{defferrard2016convolutional, kipf2017semisupervised}.
    \item A nonlinear activation function $\sigma: \mathbb{R} \to \mathbb{R}$ applied elementwise to the input signals.
    \item A local pooling or upscaling operator $\mathcal S: \R^{p' \times d_{\text{in}}} \to \R^{p'' \times d_{\text{in}}}$ that either reduces the dimensionality of the input by aggregating local information or increases it by upscaling the graph signal, together with the associated graph.
\end{itemize}

Using these building blocks, we construct a neural network, i.e., a nonlinear mapping $\mathcal{F}_{\theta_2}: \R^p \to \R^p$, such that  
\begin{equation}\label{eq:NET_generic}
\mathcal{F}_{\theta_2} := \mathcal A_{\theta_2^{(L)}} \circ \mathcal S \circ \sigma  \circ \dots \circ \mathcal S \circ \sigma \circ \mathcal A_{\theta_2^{(1)}},
\end{equation}
where each block is defined so that the output space of one block matches the input space of the next and $\theta_2^{(j)}$ denotes the $j$th block of the vector $\theta_2$. With a slight abuse of notation, as the formulation in equation \eqref{eq:NET_generic} does not define skip connections, our exact formulation of $\mathcal{F}_{\theta_2}$ is the Graph U-Net architecture proposed in \textcite{gao2022graph}. Other architectures, such as 2D convolutional neural networks, could also be adopted depending on the specific inverse problem setting considered.

The overall architecture is then defined as
\begin{equation*}
\mathcal{R}_{\theta,\tilde{\mathcal K}}
:= \mathcal F_{\theta_2} \circ \mathcal B_{\theta_1,\tilde{\mathcal K}},
\end{equation*}
where the network parameters $\theta = (\theta_1,\theta_2)$ are learned from a training set $\{(y_i,x_i) : i = 1,\ldots,n\}$ by gradient-based optimization of the objective function in equation~\eqref{eq:loss}, namely
\[
\frac{1}{n} \sum_{i=1}^n \ell\!\left(x_i,\mathcal R_{\theta,\tilde{\mathcal K}}(y_i)\right).
\]
Once the model has been trained and the optimal parameters $\hat{\theta}$ have been estimated, the source signal $x_0$ associated with a measured sensor signal $y_0$ is reconstructed as
\[
\hat{x}_0 = \mathcal R_{\hat{\theta},\tilde{\mathcal K}}(y_0).
\]

The proposed architecture extends naturally to the setting of sample-specific forward models $\tilde{\mathcal K}_i$. In that case, the network parameters $\theta = (\theta_1,\theta_2)$ are learned from training examples $\{(y_i,x_i,\tilde{\mathcal K}_i) : i = 1,\ldots,n\}$ by gradient-based optimization of the modified objective
\[
\frac{1}{n} \sum_{i=1}^n \ell\!\left(x_i,\mathcal R_{\theta,\tilde{\mathcal K}_i}(y_i)\right).
\]
Once the model has been trained, the source signal $x_0$ associated with a measured sensor signal $y_0$ and forward model $\tilde{\mathcal K}_0$ is reconstructed as
\[
\hat{x}_0 = \mathcal R_{\hat{\theta},\tilde{\mathcal K}_0}(y_0).
\]

We propose to train the model using gradient-based optimization. A key component of this approach is defining the gradient of the objective function with respect to the learnable parameters. In particular, we derive an efficient computational form for $\frac{\partial \ell}{\partial \mathcal{P}_{\theta_1}}$, as the other gradients involved have well-established forms. We defer this derivation to Section~\ref{sec:discretization}.



\subsection*{Multi-head biophysics-informed module}\label{sec:extensions}

\begin{figure}%
\centering
\includegraphics[width = 0.45\textwidth]{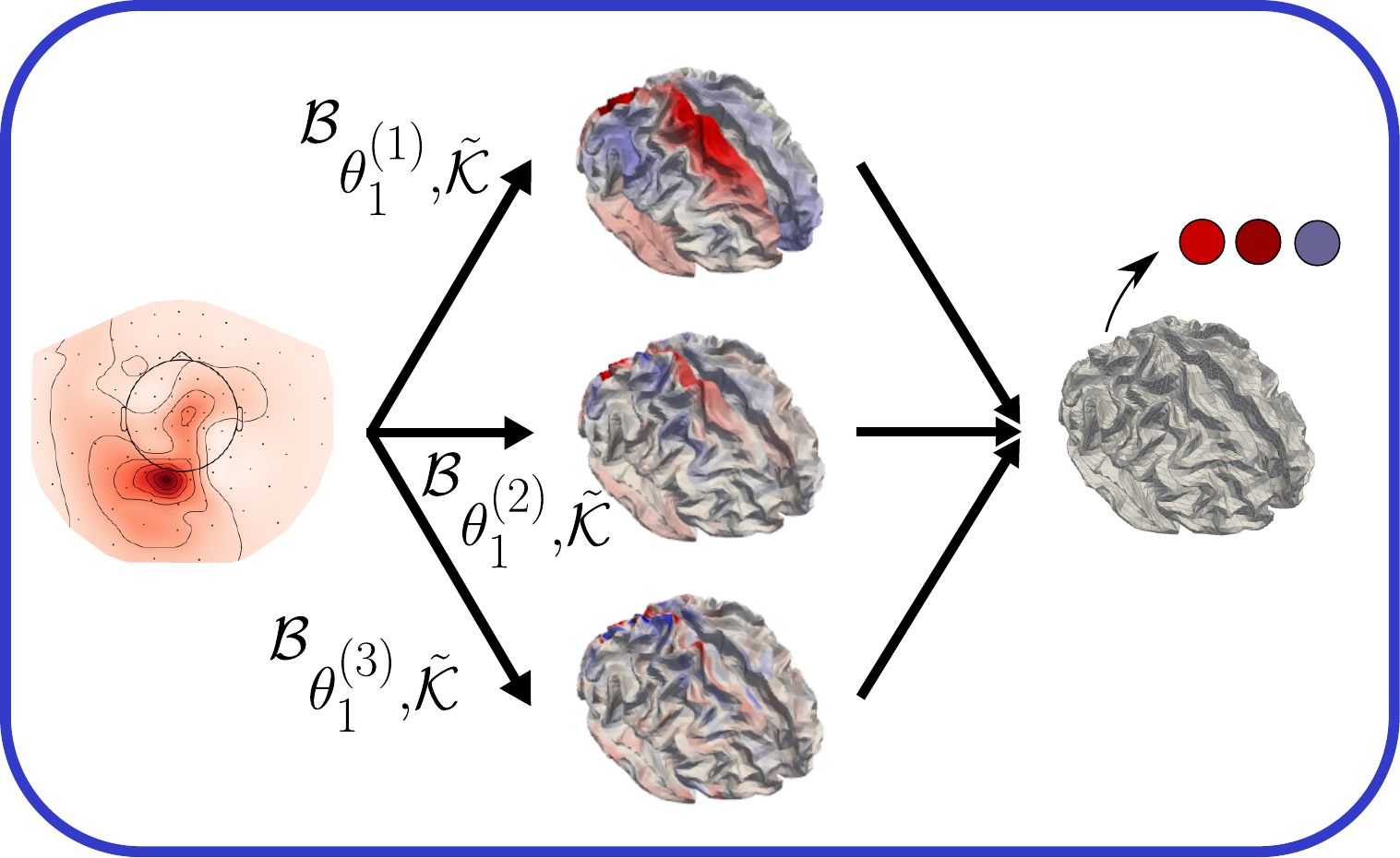}
\caption{Illustration of the proposed multi-head extension of DeepOp-Informed, obtained by stacking multiple biophysics-informed layers in parallel, each with its own set of learnable parameters. The outputs of these heads are concatenated and used as input to a graph neural network, so that each node receives a feature vector whose dimension equals the number of heads. This design can improve reconstruction performance by increasing the expressiveness of the biophysics-informed component.
}
\label{fig:extensions}
\end{figure}%

The presented biophysics-informed layer can serve as a building block for more sophisticated modules. As shown in Figure~\ref{fig:extensions}, one possible extension uses multiple heads in parallel, whose outputs are concatenated to define a multi-head biophysics-informed layer, enabling reconstruction with different learned regularization patterns, which can then be merged by the refinement module. Specifically, the following multi-head extension could be employed:  
\[
y \mapsto \left[ \mathcal{B}_{\theta_{1}^{(1)}, \tilde{\mathcal{K}}}(y), \ldots, \mathcal{B}_{\theta_{1}^{(H)}, \tilde{\mathcal{K}}}(y) \right], \qquad \theta_1 = \left(\theta_{1}^{(1)}, \ldots, \theta_{1}^{(H)}\right)
\]
where each head captures complementary information with different learnable parameters $\theta_{1}^{(h)}$. These representations are treated as separate input features by the refinement module and nonlinearly combined to produce the final reconstruction, potentially improving reconstruction performance. This extension can be used both in the fixed-forward-model setting and in the sample-specific forward-model setting.

\subsection{Gradient-based optimization}\label{sec:discretization}
We first provide a computationally efficient characterization of the forward pass, that is, the computation of the output of the biophysics-informed layer from the MEG signal, through the following proposition:

\begin{proposition}\label{prop:forward}
For a given input signal $y_i \in \mathbb{R}^s$ and parameter vector $\theta_1$, the output of the biophysics-informed layer $o_i = \mathcal{B}_{{\theta_1}, \tilde{\mathcal{K}}} (y_i) \in \mathbb{R}^{p}$ is given by the last $p$ entries of the solution to the following linear system:
\begin{align}\label{eq:forward}
\begin{cases}
\text{Find } \tilde{o}_i \text{ such that }\\
L_{\theta_1, \tilde{\mathcal{K}}} \tilde{o}_i = \tilde{y}_i,
\end{cases}
\end{align}
where
\begin{equation}
\tilde{y}_i = \begin{bmatrix} 0_{s \times 1} \\
-\tilde{\mathcal{K}}^T y_i \end{bmatrix}, \quad 
L_{\theta_1, \tilde{\mathcal{K}}} = 
\begin{bmatrix} 
I_s & \tilde{\mathcal{K}} \\ 
\tilde{\mathcal{K}}^T & -\mathcal{P}_{\theta_1}
\end{bmatrix},
\end{equation}
with $I_s$ denoting an $s \times s$ identity matrix and $0_{s \times 1}$ denoting an $s$-dimensional zero vector.
\end{proposition}

Equation~(\ref{eq:forward}) provides a reformulation of the forward pass. When $s \ll p$, as in the MEG application, this allows us to leverage the sparsity of $L_{\theta_1, \tilde{\mathcal{K}}}$ using efficient linear system solvers. Furthermore, if the forward models are identical across samples and a direct method is used to solve the linear system, matrix factorization can be performed once and reused to compute the outputs for all inputs. If an iterative method is employed, preconditioning can also be applied. Both approaches result in reduced overall computation time.

To apply gradient-based minimization, it is necessary to compute the gradients of the objective function with respect to the unknown parameters. Given that backpropagation is implemented for most standard operations in automatic differentiation packages, we focused exclusively on providing explicit expressions for $\frac{\partial \ell}{\partial \mathcal{P}_{\theta_1}}$, which is the term involving $\mathcal{B}_{{\theta_1}, \tilde{\mathcal{K}}}$. The following proposition details this computation.

\begin{proposition}\label{prop:backward}
Let $\bar{y}_i = \tilde{\mathcal{K}}^T y_i$, the output of the biophysics-informed layer $o_i = \mathcal{B}_{{\theta_1}, \tilde{\mathcal{K}}} (y_i)$, and $\frac{\partial \ell}{\partial o_i}$ be given. Then, the partial derivative $\frac{\partial \ell}{\partial \calP_{\theta_1}}$ is given by:
\begin{align}
\frac{\partial \ell}{\partial \calP_{\theta_1}} = - \sum_{i=1}^n \frac{\partial \ell}{\partial \bar{y}_i} \otimes o_i,
\end{align}
where $\otimes$ denotes the outer product and $\frac{\partial \ell}{\partial \bar{y}_i}$ 
is given by the last $p$ entries of the solution to the following linear system:
\begin{align}
\begin{cases}
\text{Find } \widetilde{\frac{\partial \ell}{\partial \bar y_i}} \text{ such that }\\
L_{\theta_1, \tilde{\mathcal{K}}} \widetilde{\frac{\partial \ell}{\partial{\bar{y}_i}}} = \widetilde{\frac{\partial \ell}{\partial o_i}},
\end{cases}
\end{align}
where $\widetilde{\frac{\partial \ell}{\partial o_i}} = \begin{bmatrix} 0_{s \times 1} \\
-\frac{\partial \ell}{\partial o_i} \end{bmatrix}$.
\end{proposition}%
This proposition shows that, given $\frac{\partial \ell}{\partial o_i}$, obtained via backpropagation from the subsequent layers, computing $\frac{\partial \ell}{\partial \bar{y}_i}$ amounts to computing the forward pass linear system in equation~(\ref{eq:forward}), using $\frac{\partial \ell}{\partial o_i}$ in place of $\tilde{\mathcal K}^T y_i$. Furthermore, once $\frac{\partial \ell}{\partial \bar{y}_i}$ is obtained, the gradient $\frac{\partial \ell}{\partial \mathcal{P}_{\theta_1}}$ reduces to a simple outer product involving $\frac{\partial \ell}{\partial \bar{y}_i}$ and $o_i$. This illustrates how both the forward and backward passes can leverage sparse linear systems to achieve more efficient computation. The filter $\mathcal{B}_{\theta_1, \tilde{\mathcal{K}}}$ was implemented in PyTorch as a custom differentiable operation via \texttt{torch.autograd.Function} \parencite{paszke2017automatic} with its forward and backward passes defined using the results derived above, enabling the biophysics-informed layer to be used in the same way as any other PyTorch layer.

Finally, for the refinement module $\mathcal{F}_{\theta_2}$, we adopted the Graph U-Net architecture implemented in PyTorch Geometric \parencite{gao2022graph}, with the graph defined by the vertices $V$ of the brain surface and edges $E$ forming the triangular faces $F$.

\printbibliography

\newpage
\appendix
\appendixpage
\setcounter{page}{1}
\setcounter{equation}{0}
\renewcommand{\theequation}{S\arabic{equation}}
\setcounter{figure}{0}
\let\oldthefigure\thefigure
\renewcommand{\thefigure}{S\oldthefigure}

\section{Additional simulation results}\label{sec:additional_sim_results}

\begin{figure}[H]
    \centering
    \includegraphics[width=1\linewidth]{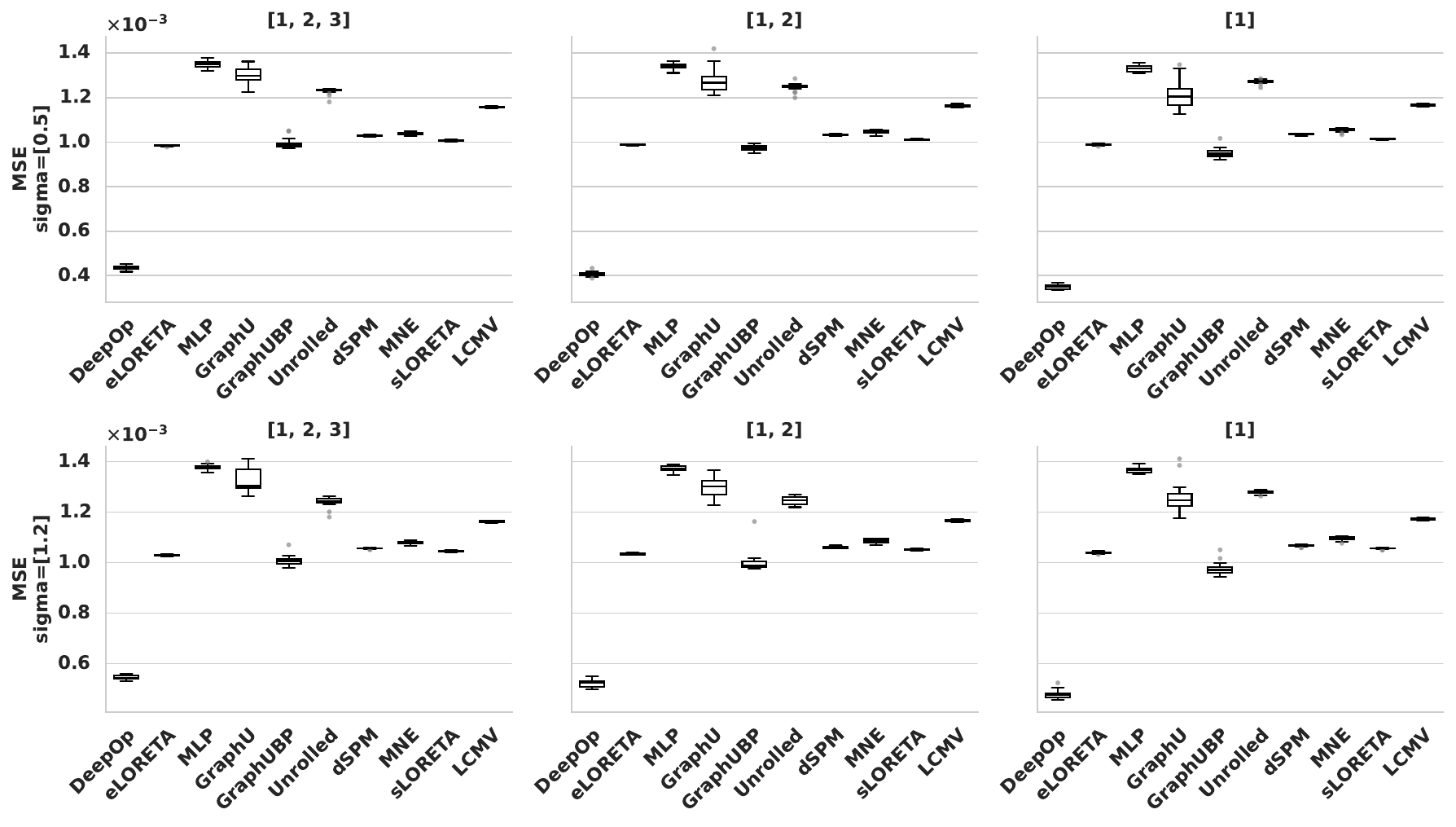}
    \caption{Boxplots of the \textit{average} test reconstruction error for the different reconstruction methods across 20 Monte Carlo experiments in the multi-subject setting. In each experiment, 200 samples were used for training and 1,000 for evaluation. The source data were generated on a reference cortical surface represented by a triangular mesh subsampled to 1,284 vertices, whose vertices served as dipoles oriented normal to the surface. Source and sensor signals were generated as described in equation~(\ref{eq:multi_subj_sim}), using for each signal a forward model randomly selected from those of 13 subjects for the training data and 6 different subjects for the test data. We used 204 gradiometers as sensors. Results are shown for two SNR levels, corresponding to Gaussian noise with standard deviations of 0.5 and 1.2. Performance was evaluated for signals with a single active source and for mixtures with $\{1,2\}$, and $\{1,2,3\}$ active sources. DeepOp-Informed outperformed all baseline methods.}
    \label{fig:realistic_subject_specific}
\end{figure}

\begin{figure}[H]
    \centering
    \includegraphics[width=1\linewidth]{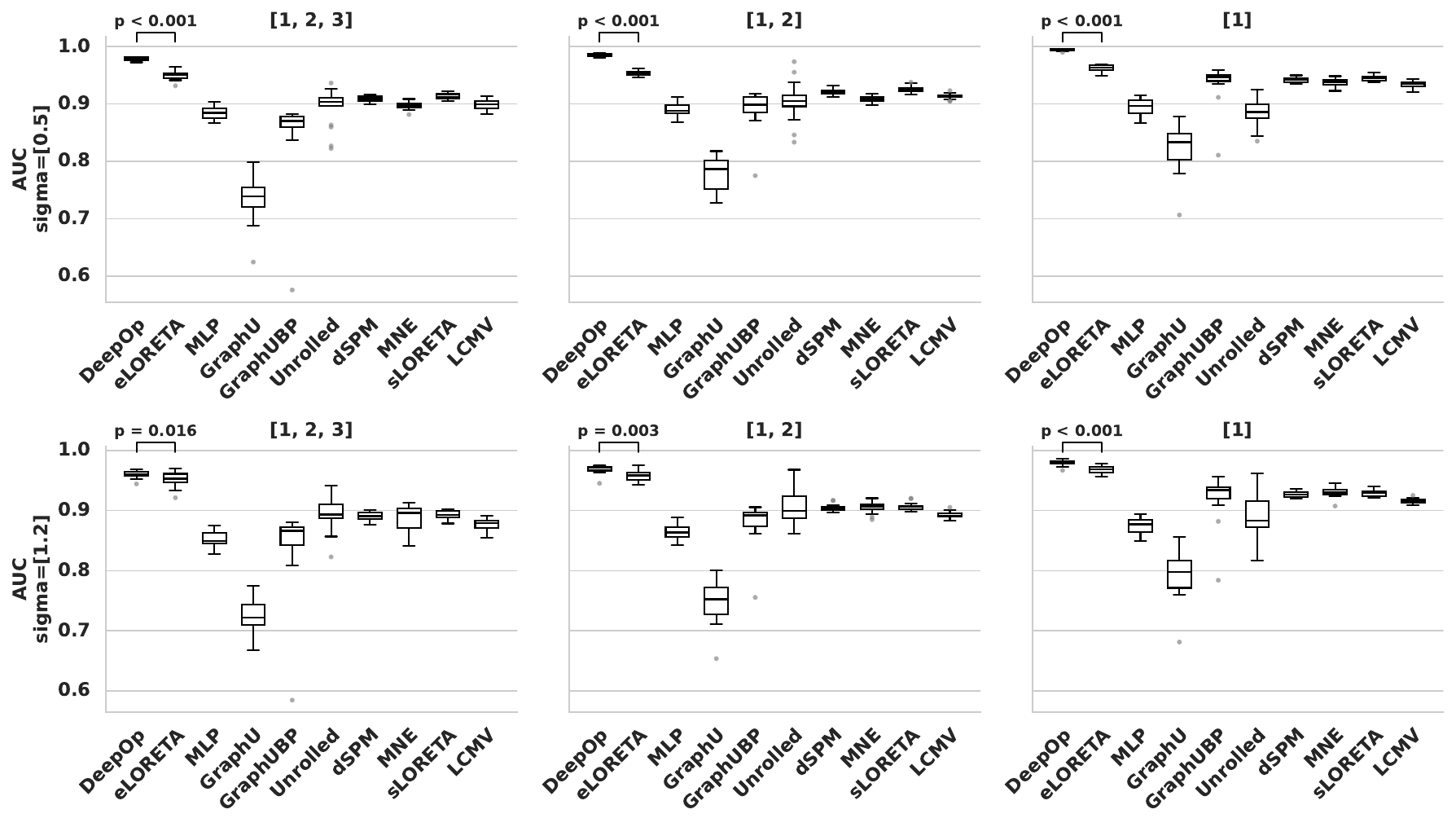}
    \caption{Boxplots of the \textit{average} test area under the receiver operating characteristic curve (AUC) achieved by the different reconstruction methods over 20 Monte Carlo experiments in the multi-subject setting shown in Figure~\ref{fig:realistic_subject_specific}. To compute the AUC, the true source signal is normalized to have a maximum value of 1, and source locations with values greater than 0.9 are labeled as active. The corresponding reconstructed values are then used as prediction scores. The AUC therefore provides a threshold-independent measure of each method’s ability to distinguish active source regions from inactive regions across the full range of possible decision thresholds. Although the improvements achieved by DeepOp over the baseline methods are smaller than those observed for normalized MSE, they remain statistically significant. This indicates that the improvements in normalized MSE cannot be attributed solely to better support recovery of the unthresholded reconstructions.}
    \label{fig:realistic_subject_specific_AUC}
\end{figure}

\begin{figure}[H]
    \centering
    \includegraphics[width=1\linewidth]{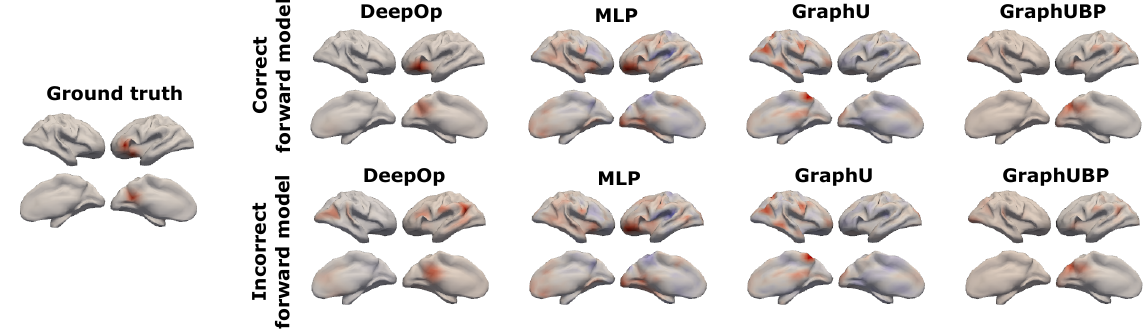}
    \caption{\textit{Incorrect forward model.} In the same setting as Figure~\ref{fig:simulation} of the main manuscript, we present example reconstructions obtained when the models were provided either with the correct forward model or with the forward model from a different subject. As expected, the reconstructions produced by MLP and GraphU are identical across the two cases, since these models do not explicitly use the forward model to generate source estimates. In contrast, the DeepOp-Informed reconstruction changes more substantially, showing that the model actively leverages subject-specific information contained in the forward model. GraphUBP, although it depends on the subject-specific forward model through the back-projected signal $\tilde{\mathcal K}_i^T y_i$, shows only minor differences between the two cases. This suggests that it may be less effective at identifying and exploiting subject-specific aspects of the forward model to improve reconstruction precision.}
    \label{fig:Multi-subj-wrong}
\end{figure}

\begin{figure}[H]
    \centering
    \includegraphics[width=1\linewidth]{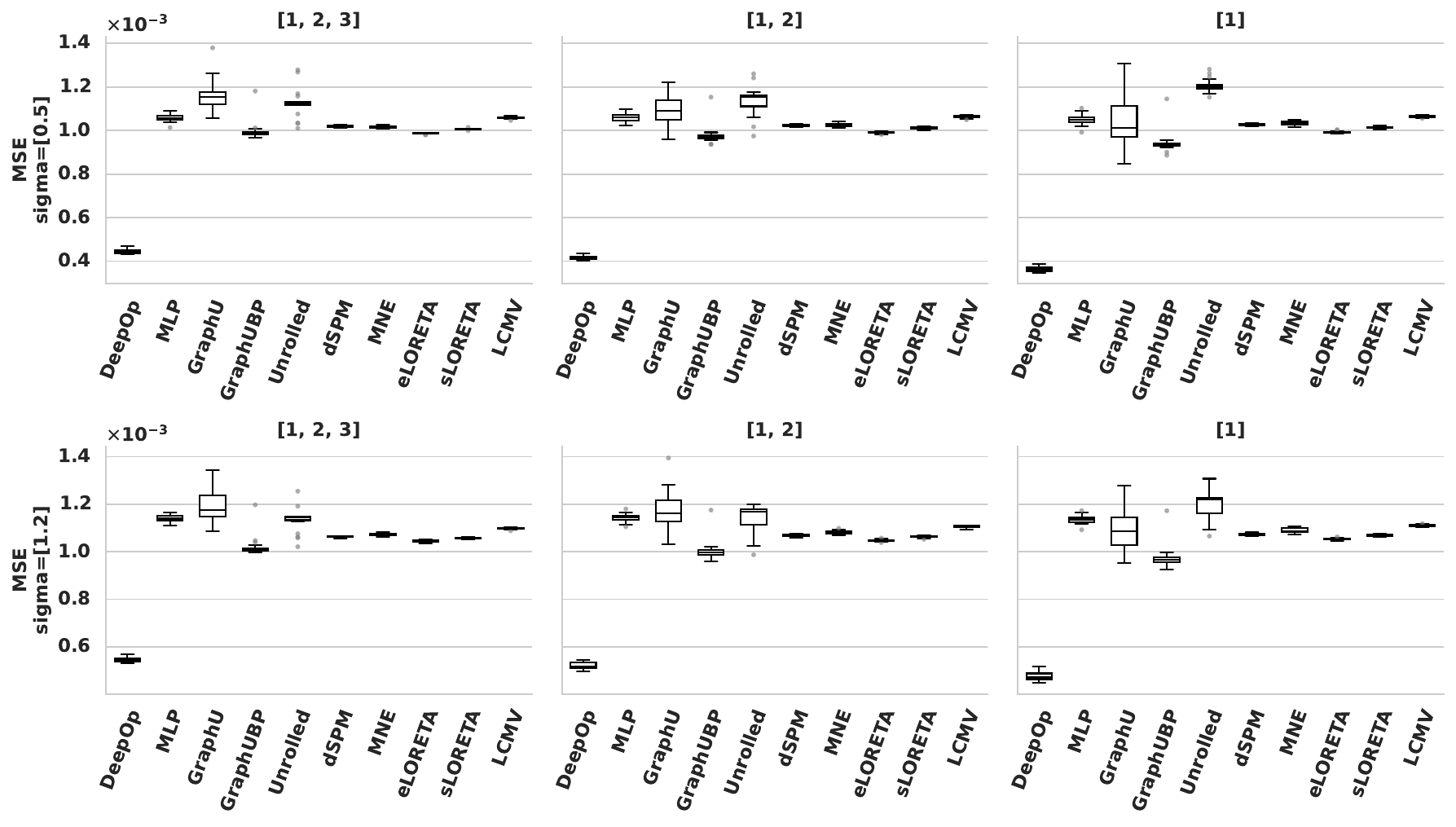}
    \caption{\textit{Realistic simulations, single-subject.} Boxplots of the \textit{average} test reconstruction error for the different reconstruction methods across 20 Monte Carlo experiments. In each experiment, 200 samples were used for training and 1,000 for evaluation. The source data were generated on the subject's cortical surface represented by a triangular mesh subsampled to 1,284 vertices, whose vertices served as dipoles oriented normal to the surface. Source and sensor signals were generated as described in equation (\ref{eq:single_subj_sim}) using a fixed forward model. We used 204 gradiometers as sensors. Results are shown for two SNR levels, corresponding to Gaussian noise with standard deviations of 0.5 and 1.2. Performance was evaluated for signals with a single active source and for mixtures with $\{1,2\}$ or $\{1,2,3\}$ active sources. DeepOp-Informed outperformed all baseline methods.}
    \label{fig:realistic_single_subject}
\end{figure}

\begin{figure}[!htb]%
\centering
\includegraphics[width = 0.25\textwidth]{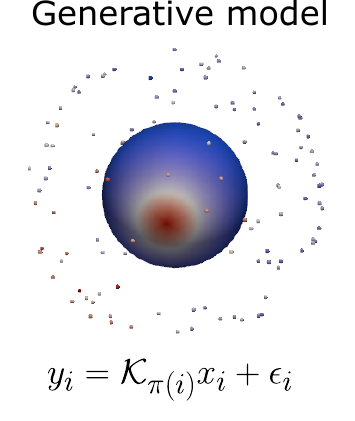}
\caption{\textit{Spherical source simulation setting.} Schematic of the data generation process. Source signals on the sphere were generated as a random number of equal-intensity localized sources centered at randomly selected vertices on the source mesh and then smoothed on the mesh to produce spatially localized but more diffuse peaks. The resulting signals were mapped through a realistic MEG forward model to 128 sensor locations outside the sphere, and i.i.d. noise was added to the simulated sensor data. The task is to reconstruct the source signals $x_i$ from the noise-contaminated measurements $y_i$.}
\label{fig:sphere-simulations}
\end{figure}%

\begin{figure}[H]
    \centering
    \includegraphics[width=0.65\linewidth]{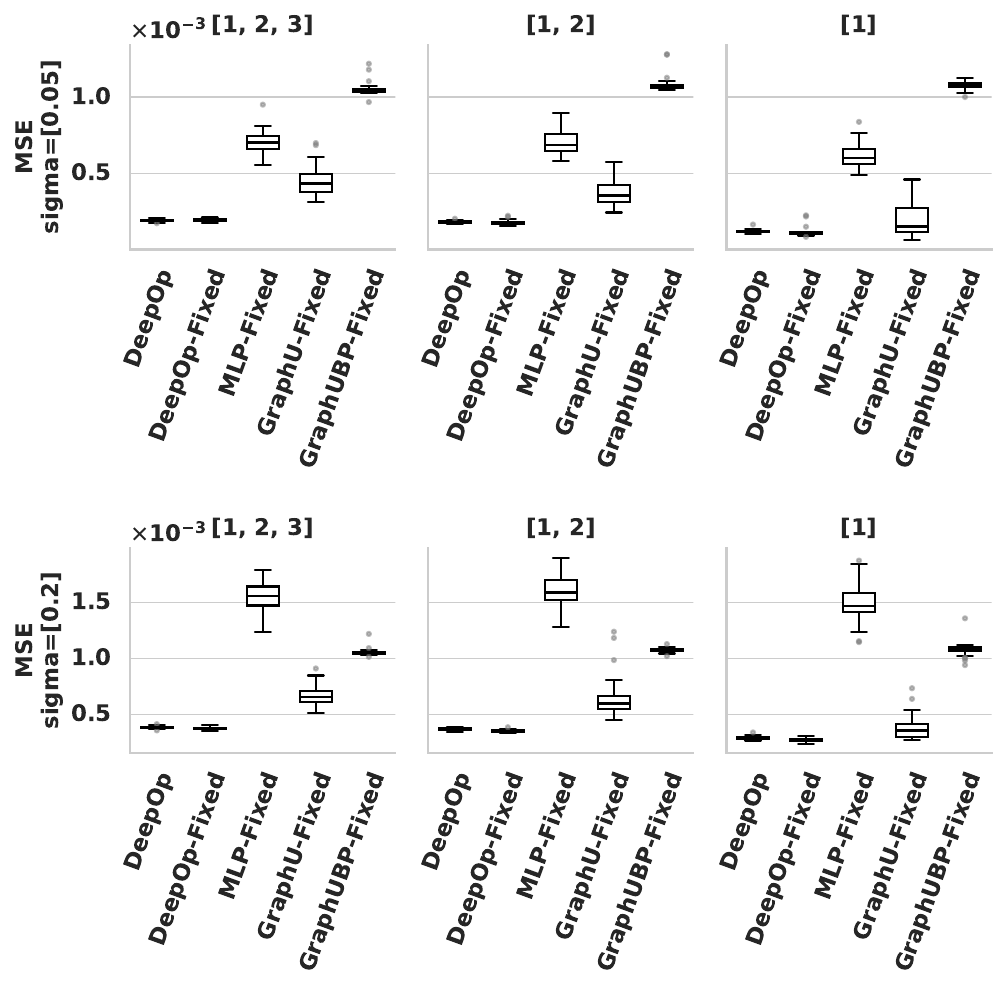}
    \caption{\textit{Spherical source, subject-specific setting.} Boxplots of the average test reconstruction error of the multi-subject version of DeepOp-Informed (DeepOp) over 50 Monte Carlo experiments in the setting described in Section~\ref{sec:additional_simulations}. In each experiment, 200 samples were used for training and 1,000 for evaluation. Training and test data were generated using different sets of forward models. Results are compared with MLP, GraphU, GraphUBP, and DeepOp-Informed trained and tested under a \textit{single-subject simulation setting}, that is, without subject-specific variation in the forward model. The comparison shows that the multi-subject model nearly matches the performance of the single-subject DeepOp-Informed model, despite the added challenge of generalizing across subjects with subject-specific forward models. Results are shown for two SNR levels, corresponding to Gaussian noise with standard deviations of 0.05 and 0.2. Performance was evaluated for signals with a single active source and for mixtures with $\{1,2\}$, and $\{1,2,3\}$ active sources.}
    \label{fig:sphere-subject-specific}
\end{figure}

\begin{figure}[H]
    \centering
    \includegraphics[width=0.7\linewidth]{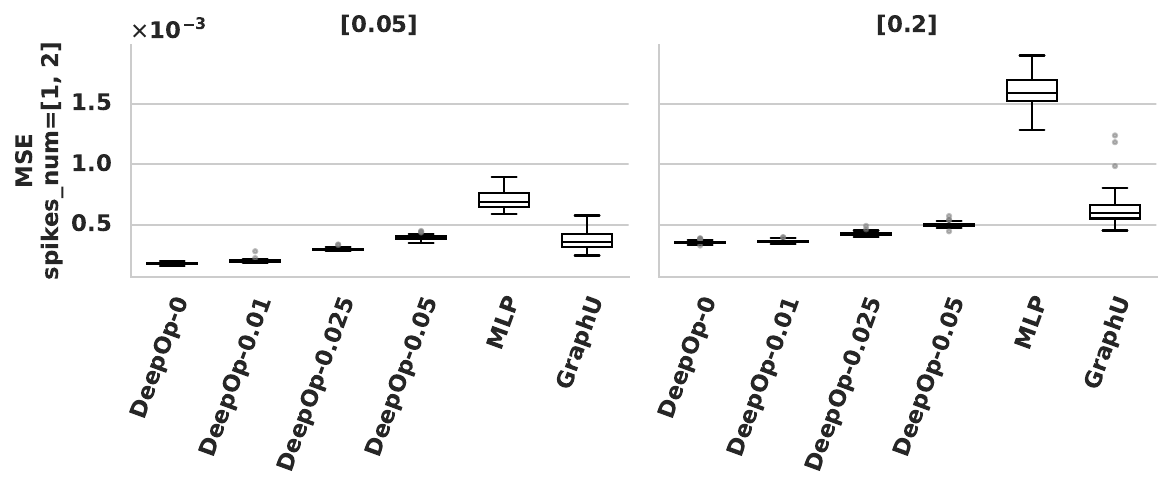}
    \caption{\textit{Spherical source, forward model perturbation.} Boxplots showing the average test reconstruction error of DeepOp-Informed over 50 Monte Carlo experiments, when the model is informed with a perturbed version $\tilde{\mathcal{K}}$ of the forward model $\mathcal{K}$ used for data generation. In each experiment, 200 samples were used for training and 1,000 for evaluation. Perturbations were applied independently to each entry of the forward model $\mathcal{K}$ used for data generation and were drawn from a uniform distribution $\mathrm{Unif}(0, c \, k_{\max})$, where $k_{\max}$ is the largest element of the forward model $\mathcal{K}$, and $c$ is set to 0 (no perturbation), 0.01, 0.025, and 0.05. The matrix $\tilde{\mathcal{K}}$ was then used to inform the proposed model. Here, we ran experiments only for source signals consisting of mixtures of 1 or 2 active sources. DeepOp-Informed is able to leverage the forward model information even when the provided model is perturbed within a reasonable range. However, under higher SNRs, its performance degrades more rapidly with increasing perturbation magnitude compared to the low SNR setting.}
    \label{fig:sphere-perturb}
\end{figure}

\begin{figure}[H]
    \centering
    \includegraphics[width=0.9\linewidth]{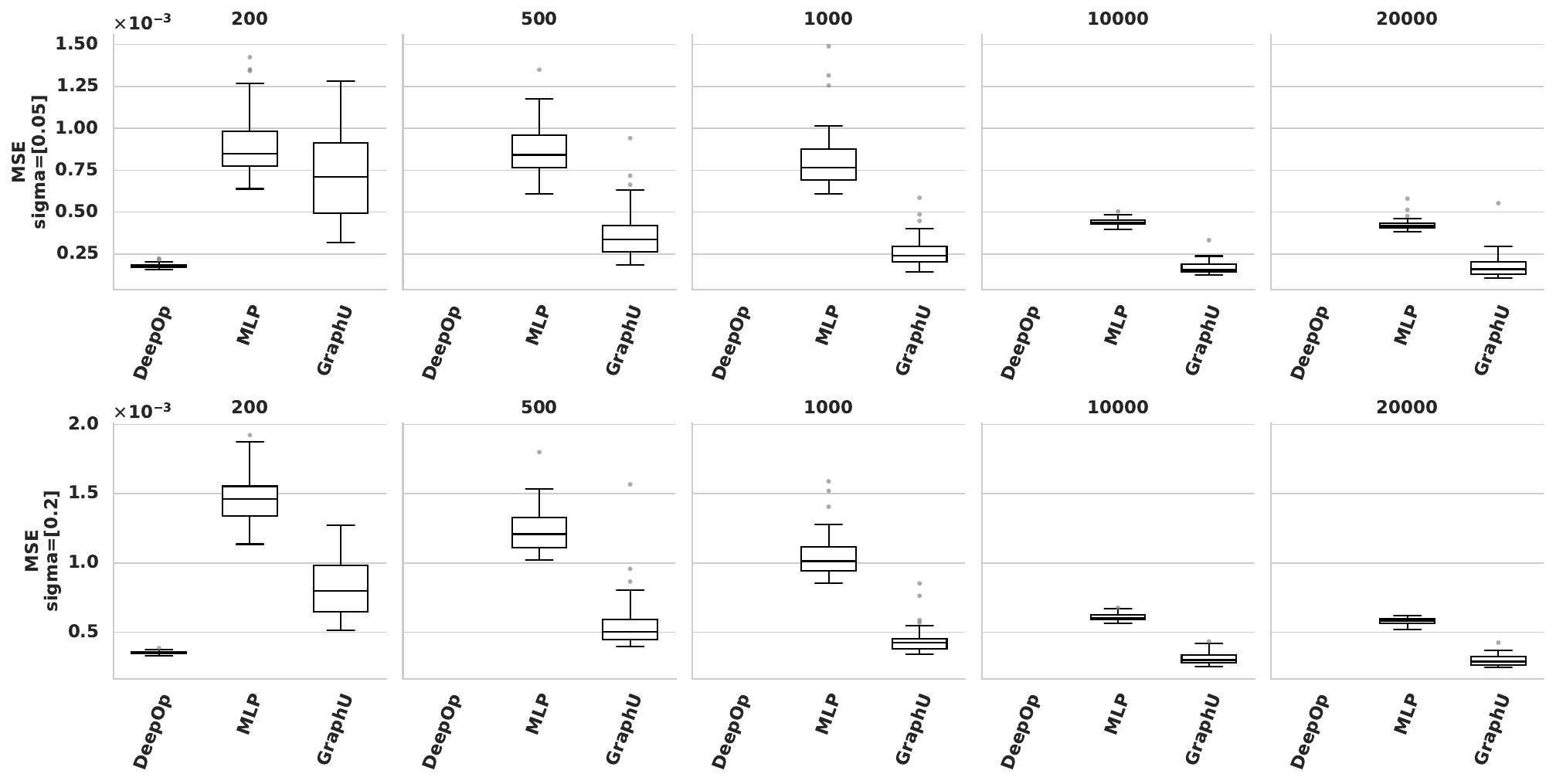}
    \caption{\textit{Spherical source, sample size.} Boxplots showing the test reconstruction error of MLP and GraphU over 50 Monte Carlo experiments for increasing sample sizes under the basic setting. As expected, the performance of MLP and GraphU improves with increasing sample size. However, MLP does not match the performance of DeepOp-Informed even with 20,000 training samples. GraphU requires more than 5,000 samples in the low-SNR setting and more than 10,000 in the high-SNR setting to reach the performance of DeepOp-Informed trained on only 200 samples.}
    \label{fig:sphere-sample-size}
\end{figure}

\begin{figure}[H]
    \centering
    \includegraphics[width=0.65\linewidth]{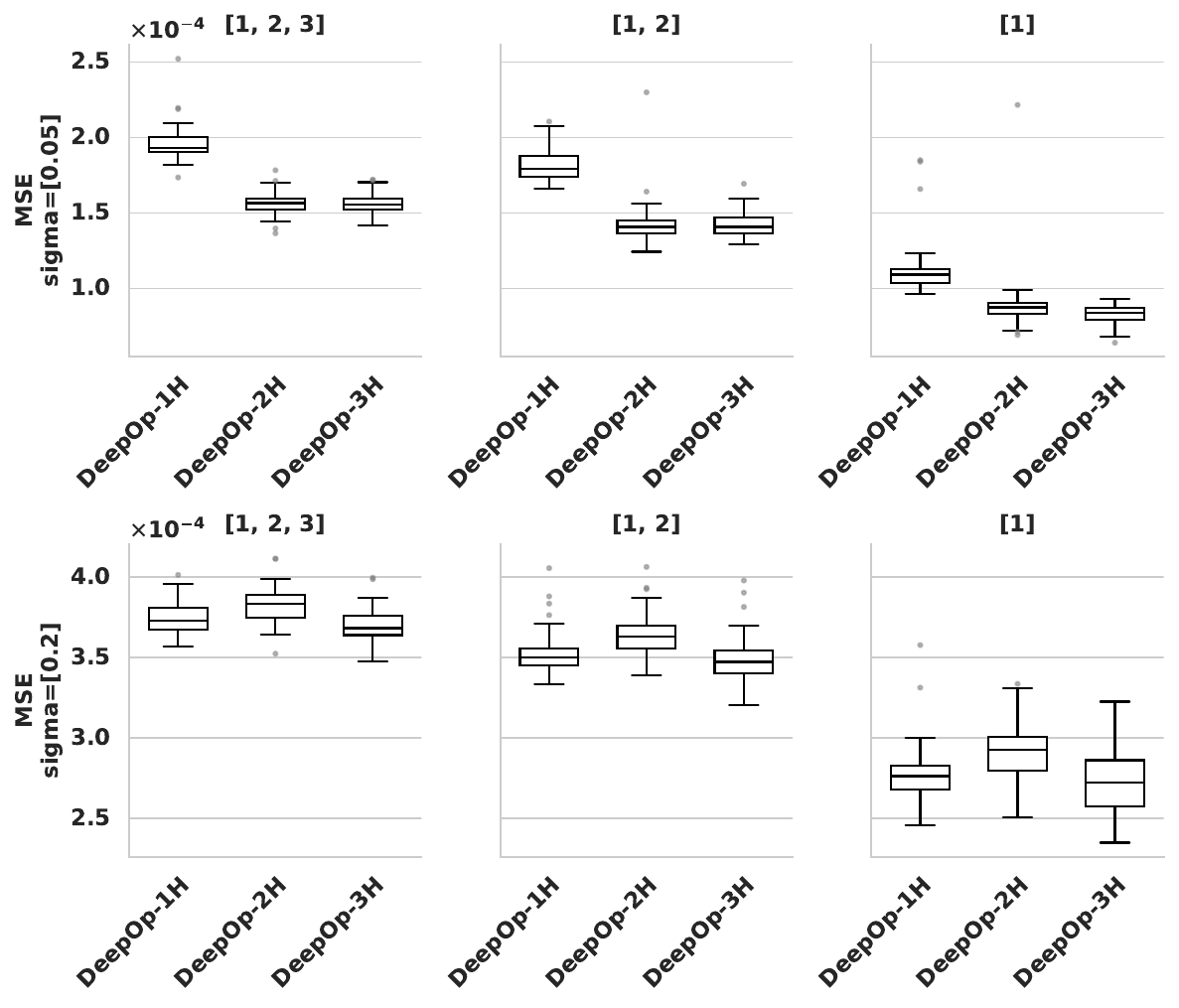}
    \caption{\textit{Spherical source space, multi-head extension.} Boxplots showing the test reconstruction error of DeepOp-Informed ({DeepOp-1H}) and its two- and three-head extensions ({DeepOp-2H} and {DeepOp-3H}) over 50 Monte Carlo experiments in the spherical source-space setting with a fixed forward model (Section~\ref{sec:extensions}). In the higher-SNR setting, the multi-head architecture appears to improve performance. In the lower-SNR setting, the results are more mixed. The number of heads should therefore be treated as a tuning parameter.
}
    \label{fig:sphere-multi-head}
\end{figure}

\begin{figure}[H]
    \centering
    \includegraphics[width=1\linewidth]{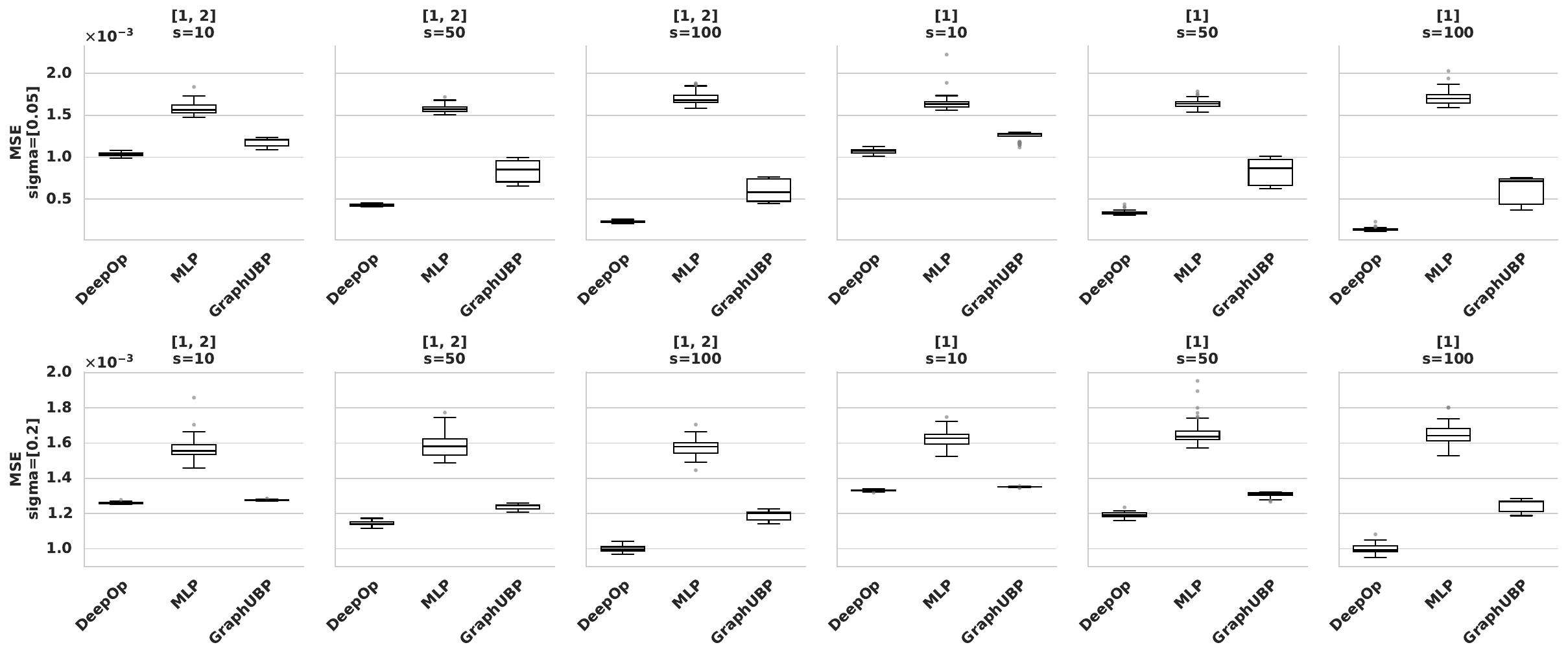}
    \caption{\textit{Spherical source space, direct measurements.} Boxplots showing the average test reconstruction error in the sample-specific setting, where the signal is measured directly in source space but only at randomly sampled locations, over 50 Monte Carlo experiments. In each experiment, 200 samples were used for training and 1,000 for evaluation. Here, the ``forward models'' $(\mathcal{K}_i)$ are $s \times 642$ matrices, with $s = 10, 50, 100$ across different simulations. Each $\mathcal{K}_i$ was constructed to have all entries equal to zero except for a single randomly chosen entry in each row, which is set to 1. This emulates the process of sampling the source signal at $s$ random subject-specific locations, with only these sampled values available for reconstruction. The training and test sets were generated using different sets of forward models. Results are compared with those of MLP and GraphUBP trained on the same datasets. Results are shown for two SNR levels, corresponding to Gaussian noise with standard deviations of 0.05 and 0.2. Performance was evaluated for signals with one peak and for mixtures of one and two peaks. DeepOp-Informed outperformed the baseline method GraphUBP (that is, GraphU applied to back-projected data). The MLP results are included for completeness, but this model cannot perform better than chance because its input variables are permuted unpredictably from sample to sample.
    }
    \label{fig:sphere-missing-data}
\end{figure}

\section{Evaluation on MEG recordings: Additional results}\label{sec:additional_results}

\begin{figure}[H]
    \centering
    \includegraphics[width=0.4\linewidth]{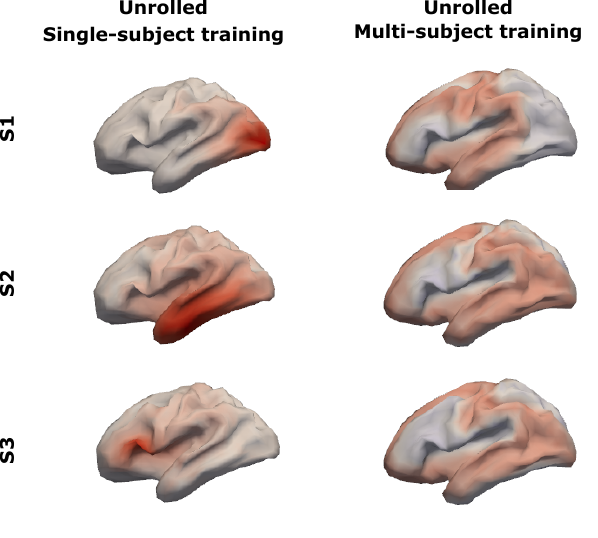}
    \caption{\textit{Source reconstruction unrolled network.} Source reconstructions of evoked responses to the first (S1), second (S2), and third (S3) syllables in the auditory experiment, obtained using the unrolled neural-network architecture. The reconstructions deviated more markedly from the expected patterns, suggesting that further architectural modifications may be needed to adapt this model to the MEG setting.}
    \label{fig:Application_unrolled}
\end{figure}

\begin{figure}[H]
    \centering
    \includegraphics[width=0.6\linewidth]{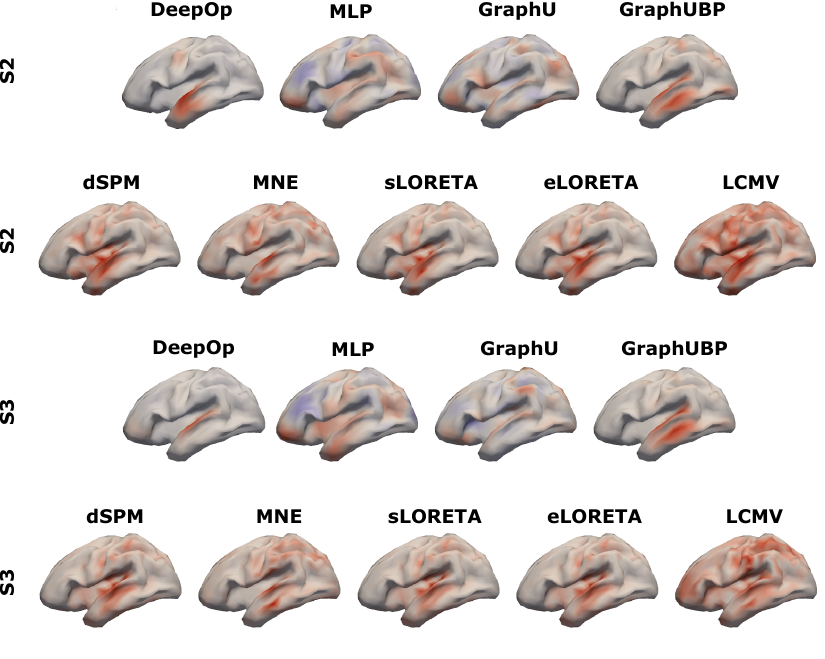}
    \caption{\textit{Stability analysis.} Source reconstructions of the evoked responses to the second (S2) and third (S3) syllables in the auditory experiment, which can be viewed as perturbations of the evoked response to the first syllable (Figure~\ref{fig:application}). DeepOp-Informed provides robust source reconstructions, while also producing sharper source estimates than the other methods.}
    \label{fig:Application_sens}
\end{figure}

\section{Proofs}\label{sec:proofs}

\begin{proof}[Proof of Proposition 1]
The problem   
\begin{equation}
o_i  = \mathcal{B}_{\theta_1, \tilde{\mathcal K}}(y_i) = \argmin_{x \in \R^p} \left\| y_i - \tilde{\mathcal{K}} x \right\|_2^2 + \left\| \mathcal{P}^{\frac{1}{2}}_{\theta_1} x \right\|^2_2,
\end{equation}
can be equivalently written as $o_i = (\tilde{\mathcal{K}}^T \tilde{\mathcal{K}} + \mathcal{P}_{\theta_1})^{-1} \tilde{\mathcal{K}}^T y_i$, assuming $\tilde{\mathcal{K}}^T \tilde{\mathcal{K}} + \mathcal{P}_{\theta_1}$ is invertible. From equation (\ref{eq:forward}), by partitioning $\tilde{o}_i=\begin{bmatrix} u_i\\ o_i\end{bmatrix}$, using the top block to get $u_i=-\tilde{\mathcal K}o_i$, and substituting into the bottom block, we obtain $(\tilde{\mathcal K}^\top\tilde{\mathcal K}+\mathcal P_{\theta_1})o_i=\tilde{\mathcal{K}}^T y_i$, hence $o_i=(\tilde{\mathcal K}^\top\tilde{\mathcal K}+\mathcal P_{\theta_1})^{-1} \tilde{\mathcal{K}}^T y_i$.
\end{proof}

\begin{proof}[Proof of Proposition 2]
Recall that the output of the biophysics-informed layer, $o_i$, is computed in the forward pass from its input $y_i$ via $o_i = \mathcal{B}_{\theta_1, \tilde{\mathcal{K}}} (y_i)$. Let $\frac{\partial \ell}{\partial o_i}$ denote the gradient of the objective with respect to $o_i$, as obtained via backpropagation from the subsequent layers of the neural network. We write the $j$-th component of $\bar y_i$ as $\bar y_{ij}$, and abbreviate $\mathcal{B}_{\theta_1, \tilde{\mathcal{K}}}$ by $\mathcal{B}$ and $\mathcal{P}_{\theta_1}$ by $\mathcal{P}$. Moreover, let $\mathcal S = \tilde{\mathcal{K}}^T \tilde{\mathcal{K}} + \mathcal{P}$ and $\bar{y}_i = \tilde{\mathcal{K}}^T y_i$, and notice that $o_i = \mathcal{B} (y_i) = \mathcal{S}^{-1}\bar{y}_i$.

Then, we have:
\begin{align*}
    \left(\frac{\partial \ell}{\partial \bar{y}_i} \right)_j &= \sum_r \frac{\partial \ell}{\partial o_{ir}} \frac{\partial o_{ir}}{\partial \bar{y}_{ij}}\\
    &= \sum_r\frac{\partial \ell}{\partial o_{ir}} \frac{\partial}{\partial \bar{y}_{ij}} \left( \sum_k \mathcal{S}_{rk}^{-1} \bar{y}_{ik} \right)\\
    &= \sum_r \frac{\partial \ell}{\partial o_{ir}} \left( \sum_k \mathcal{S}_{rk}^{-1}  \frac{\partial \bar{y}_{ik}}{\partial \bar{y}_{ij}} \right)\\
    &= \sum_r \frac{\partial \ell}{\partial o_{ir}} \left( \sum_k \mathcal{S}_{rk}^{-1}  \delta_{kj}\right)\\
    &= \sum_r \frac{\partial \ell}{\partial o_{ir}} \mathcal{S}_{rj}^{-1}.
\end{align*}
Hence,
\[
\frac{\partial \ell}{\partial \bar{y}_i} = (\mathcal{S}^{-1})^T  \frac{\partial \ell}{\partial o_i},
\]
which, together with the symmetry of $\mathcal{S}$, implies that $\frac{\partial \ell}{\partial \bar{y}_i}$ is given by the solution to
\begin{align*}
\begin{cases}
\text{Find $\frac{\partial \ell}{\partial \bar{y}_i}$ such that}\\ \mathcal S \frac{\partial \ell}{\partial \bar{y}_i} = \frac{\partial \ell}{\partial o_i}.
\end{cases}
\end{align*}
It is straightforward to show that this is equivalent to the sparse linear system in equation (\ref{eq:forward}) using $\frac{\partial \ell}{\partial o_i}$ in place of $\tilde{\mathcal K}^T y_i$.

 Next, we compute the gradient of the objective with respect to $\calP$:
\begin{align*}
\left( \frac{\partial \ell}{\partial \calP} \right)_{mn}
&= \sum_{i,r} \frac{\partial \ell}{\partial o_{ir}} \frac{\partial o_{ir}}{\partial \calP_{mn}}\\
&= \sum_{i,r} \frac{\partial \ell}{\partial o_{ir}} \frac{\partial}{\partial \calP_{mn}} (\sum_{j} \mathcal{S}_{rj}^{-1} \bar{y}_{ij})\\
&=-\sum_{i,r}\frac{\partial \ell}{\partial o_{ir}} \sum_{j} \mathcal{S}_{rj}^{-1} \sum_{k} \frac{\partial \mathcal{S}_{jk}}{\partial \mathcal{P}_{mn}} (\sum_{l} \mathcal{S}_{kl}^{-1} \bar{y}_{il})\\
&=-\sum_{i,r}\frac{\partial \ell}{\partial o_{ir}} \sum_{j} \mathcal{S}_{rj}^{-1} \sum_{k} \delta_{jm}\delta_{kn} (\sum_l \mathcal{S}_{kl}^{-1} \bar{y}_{il})\\
&=-\sum_{i,r} \frac{\partial \ell}{\partial o_{ir}}  \mathcal{S}_{rm}^{-1} (\sum_l \mathcal{S}_{nl}^{-1} \bar{y}_{il}).
\end{align*}
Therefore,
\begin{align}
\frac{\partial \ell}{\partial \mathcal{P}} 
&= -\sum_i \left((\mathcal{S}^{-1})^T \frac{\partial \ell}{\partial o_i} \right) \otimes (\mathcal{S}^{-1} \bar{y}_i)\\
&= -\sum_i \frac{\partial \ell}{\partial \bar{y}_i} \otimes o_i.
\end{align}
\end{proof}

\end{document}